\documentclass[11pt,a4paper]{article}
\usepackage[margin=25mm]{geometry}
\usepackage[T1]{fontenc}
\usepackage{lmodern,microtype,amsmath,amssymb,amsthm,booktabs,tabularx,array,graphicx,xcolor,placeins}
\usepackage[colorlinks=true,linkcolor=blue!50!black,citecolor=blue!50!black,urlcolor=blue!50!black]{hyperref}
\hypersetup{pdftitle={Collapse and Causality in Generalized Trace Dynamics},pdfauthor={Tejinder P. Singh}}
\newcommand{\Tr}{\operatorname{Tr}}
\newcommand{\E}{\mathbb E}
\newcommand{\id}{\mathbb I}
\newcommand{\dd}{\mathrm d}
\newcommand{\sinc}{\operatorname{sinc}}
\newcommand{\logit}{\operatorname{logit}}
\theoremstyle{plain}
\newtheorem{proposition}{Proposition}
\newtheorem{lemma}{Lemma}
\theoremstyle{definition}
\newtheorem{example}{Example}
\theoremstyle{remark}
\newtheorem{remark}{Remark}
\title{Collapse and Causality in Generalized Trace Dynamics\\[5pt]\large Record selection in Connes time and the causal consistency of emergent records}
\author{Tejinder P. Singh\\[3pt]
\small Tata Institute of Fundamental Research\\
\small Homi Bhabha Road, Mumbai 400005, India\\
\small\texttt{tpsingh@tifr.res.in}}
\date{October 4, 2026}
\begin{document}
\maketitle

\begin{abstract}
In generalized trace dynamics (GTD) the fluctuations proposed to select a measurement record come from unresolved pregeometric degrees of freedom and are ordered by Connes time rather than by a spacetime coordinate. The hypersurface on which collapse occurs is then not a question posed at the level at which selection happens; it returns as the requirement that the emergent records do not signal. We state this requirement operationally, apply the standard no-signalling theorem to an effective collapse law of the continuous-spontaneous-localisation type, and establish results that constrain how GTD can meet it. The norm-preserving but response-free law of a naive coarse-graining is solved exactly: it never selects a record, and for two records it erodes every initial weight towards one half; the Born-restoring correction is a drift tangent to the sphere of states, invisible to norm conservation. In history form, record weights are martingales only if the fluctuation history leans towards the leading record with gain twice its own diffusion and, for record operators with three or more collinear eigenvalues, the linear law carries its counter-term with unit coefficient; a smooth coefficient drive gives the counter-term with the wrong sign, and any other gain is an ensemble nonlinearity that signals. For a sea in canonical equilibrium the fluctuation--dissipation relation fixes the gain, so the Born gain is an Einstein relation between coupling, temperature and the coefficient's correlation time. A rate that reads the entanglement of the conditioned state signals; state dependence may enter only through the unravelling, where it suppresses selection unobservably, so observable enhancement must be carried by fixed collective operators. A split-signature geodesic length is a norm only after a positive subspace is selected. Positive states annihilate the Grassmann bilinear that carries the collapse coefficient, so the physical block is a reconstruction on which the kinetic form is positive. The coefficient noise is evaluated for a sea taken from the programme, the Sutherland texture of the relational master matrix: a block coupled to its local strain acquires an Ohmic spectral density with computed friction, white-noise weight, continuum criterion and recurrence, whereas a purely kinetic cross-coupling gives no friction. What remains is the projection that supplies the counter-term and the gain for an interacting sector.
\end{abstract}

\tableofcontents

\section{Introduction}\label{sec:intro}
\subsection{One quantum, many detectors, one record}\label{sec:puzzle}
A source emits a single quantum towards a sphere of detectors. The wavefront reaches all of them, yet an ideal experiment leaves exactly one absorption record. Unitary quantum mechanics describes how the quantum becomes correlated with the apparatus and produces a superposition of alternatives---detector~1 fired, detector~2 fired, and so on---but it does not pick one of them. This is the single-particle measurement problem in the form Einstein gave it at the 1927 Solvay conference \cite{solvay}; the detector sphere is a convenient idealisation of his argument.

It is tempting to picture the winning detector sending an instantaneous instruction to all the others not to respond. A stochastic collapse model contains no such instruction. It specifies a single joint law for the state and for the histories of the fluctuations that drive it; in each run the relative weights of the alternatives change, and one of them eventually dominates. Whether this is causally acceptable is an operational question: can a choice made at one detector alter the statistics recorded at a spacelike-separated one? Conditional probabilities change whenever information becomes available, and that by itself is not a communication channel.

\subsection{The GTD proposal: selection before spacetime}
Generalized trace dynamics (GTD) asks a prior question. It proposes that quantum theory, classical bodies and spacetime itself arise from a deterministic dynamics of noncommuting matter--geometry variables, the aikyons---the proposed atoms of space-time-matter \cite{adler,gtdcritical,book}. A detector is then not a classical object introduced from outside; localisation must account for its classical behaviour as well as for the definiteness of its readings \cite{space2018,space2019}.

The physical picture examined in this paper is the following. When a quantum system's interaction with a large number of unresolved aikyons is coarse-grained, the proposal is that what survives is a fluctuating anti-self-adjoint contribution to its effective evolution. These fluctuations are not fields on spacetime. They are indexed by Connes time $\tau$---the modular, or thermal, time of Connes and Rovelli \cite{connes}, which GTD takes as the ordering parameter of the pregeometric dynamics \cite{gtdtime}---and they couple to collective alternatives of the system, not to spacetime points. Spacetime, with its light cones, is a later description: a property of the localised configurations that the same dynamics produces. The familiar question of relativistic collapse---on which spacelike hypersurface does the reduction occur, and how can the answer be the same in every frame---is therefore not posed at the level at which selection happens, because no hypersurface exists there yet. It is deferred, not removed: once records admit a spacetime description, their statistics must respect the causal structure of that spacetime, and the emergent description must be covariant to the accuracy observed.

It is important to be clear about what this picture does and does not buy. Deferring the hypersurface question is not the same as proving no-signalling. No-signalling is a property of the effective law obeyed by the emergent records: it holds when that law is linear on ensembles and composed of locally complete operations, and it fails otherwise, whatever the fluctuations are called and wherever they live. The emergent law must therefore be exhibited and tested. This paper constructs such a law from standard ingredients, applies the no-signalling theorem to it, and shows by explicit calculation where natural-looking prescriptions fall outside its conditions. It then follows the GTD coefficient that is supposed to supply the fluctuations as far as the construction of Appendix~\ref{app:micro} allows, and states what is still missing.

\subsection{The argument in plain terms}\label{sec:plain}
The following ten steps state the content of the paper without equations; each points to the section that does the work.
\begin{enumerate}
\item \emph{The puzzle.} One quantum reaches every detector on a sphere, yet one detector clicks. Ordinary quantum mechanics does not say what happens at the moment the click is chosen, and Einstein's worry was that the choice seems to be communicated across the sphere instantaneously (Section~\ref{sec:puzzle}).
\item \emph{The usual fix and its sore spot.} Collapse models add a small random nudge to every quantum system, which over many repetitions favours one alternative until it wins; no message is sent. But relativity has no universal ``now'', so a collapse that happens ``at an instant'' looks different in every frame (Section~\ref{sec:connes}).
\item \emph{The GTD proposal.} Space and time are not fundamental. The random nudges come from unresolved pregeometric building blocks and act in an internal ordering parameter, Connes time, before space and time exist, so the question ``in whose frame does the collapse happen?'' is not posed at that level. It is deferred, not removed: once spacetime has emerged, the clicks we record must respect its light cones, and the emergent description must be covariant (Sections~\ref{sec:connes}, \ref{sec:continuum}).
\item \emph{Checking the clicks obey relativity.} For a standard collapse rule with fixed local operators, one detector wins with the frequencies quantum mechanics predicts, and Alice's statistics never depend on Bob's choice. This is the standard no-signalling theorem; nothing in it is specific to GTD, and its value here is in what it rules in and out for the deeper theory (Sections~\ref{sec:solvay}, \ref{sec:nosig}).
\item \emph{Keeping the total probability fixed is not enough.} The shortcut ``the nudges just redistribute probability, keeping the total at one'' is solved exactly: the odds between two detectors wander like a fair coin, return to even odds infinitely often, and the leading alternative loses its lead on average. A real collapse needs an extra steady push toward the leader, and that push is invisible to any bookkeeping of total probability (Proposition~\ref{prop:bare}, Section~\ref{sec:normdrift}).
\item \emph{The push, seen from the side of the nudges.} The same requirement can be read off the random history itself: the history must lean toward the leading alternative with a definite gain, exactly twice its own jitter, and the law must carry one fixed correction term besides. A history that leans with half that gain, or twice it, still ends in one record but with the wrong frequencies, and such a law would let distant observers signal. If the history is produced by a surrounding medium in thermal balance, the lean is that medium's response to the system, and the relation between its jitter and its response---the fluctuation--dissipation relation---fixes the gain: the Born gain is an Einstein relation between coupling, temperature and the medium's memory time. The correction term is a separate demand that the medium's response does not supply, and a finite memory time makes the collapse law one with memory (Propositions~\ref{prop:history}, \ref{prop:einstein}; Sections~\ref{sec:history}, \ref{sec:einstein}).
\item \emph{Where ``faster collapse for bigger systems'' may come from.} If the collapse rate reads off how entangled the whole system is, Bob can signal Alice. The rate may instead depend on the state in a way no experiment can see, but then it can only slow collapse down. Any speed-up that is both admissible and observable must come from the size of what is being distinguished, as in existing collapse models (Propositions~\ref{prop:ent}, \ref{prop:bound}, Example~\ref{ex:amplification}).
\item \emph{The geometry does not hand over probabilities.} The six-dimensional form of the programme has three time-like and three space-like directions, so a length can be positive, negative or zero. A conserved positive length exists on a cone of velocities, but a probability norm must be positive on a linear space of states, and that requires choosing a positive subspace; the choice is of the same kind as choosing a four-dimensional spacetime inside the six, which is why norm conservation cannot be an input to that choice (Section~\ref{sec:geodesicnorm}). The six dimensions do no work in the no-signalling argument of step~4.
\item \emph{The noise can be there, and the programme's own medium can supply it.} In the free theory the proposed source of the nudges does not fluctuate at all. Coupled to a continuum of unresolved neighbours it jitters as a collapse rule needs, with a computed strength, and a finite set of neighbours makes it echo back instead. The surrounding medium is then taken from the programme itself: the ordered texture of the relational master matrix is a chain of neighbours whose vibration frequencies form a continuous band when the chain is taken long at fixed spacing (a choice the programme has not made), and a system attached to it by an inserted coupling acquires friction and jitter of the required kind, with the number of neighbours needed and the time after which the echo returns both computed. Two cautions come with this: the block whose kinetic energy is the indefinite collapse coefficient admits only one state, in which that coefficient is invisible, so the physical block is a reconstructed one on which the kinetic energy is positive; and a coupling through velocities alone gives no friction, so the interaction must act on the block's position (Lemma~\ref{lem:positive}, Proposition~\ref{prop:bath}, Example~\ref{ex:sutherland}, Section~\ref{sec:bath}).
\item \emph{What is still missing.} One calculation: start from the interacting theory, follow a system with its unresolved neighbours, and show that what comes out is a collapse rule with the push of step~5 at the gain of step~6, the operator-carried size dependence of step~7, the positive norm of step~8 and the noise of step~9 (Section~\ref{sec:closure}).
\end{enumerate}

\subsection{Results}
The paper combines an effective analysis, which uses established collapse and quantum-operations theory, with microscopic tests that use the GTD equations of Appendix~\ref{app:micro} directly. Its results are the following.
\begin{enumerate}
\item \emph{The bare norm-preserving law never selects (Proposition~\ref{prop:bare}).} The law obtained by centring a fluctuating gain so that the norm is conserved, with no further response, is the normalisation of the linear law driven by centred noise, with no counter-term and no reweighting. Solved exactly, every log-ratio of record weights is a driftless Brownian motion: it returns to even odds infinitely often, so no record is ever permanently selected, and for two records the mean weight erodes to one half from every initial value. That a norm-preserving law without the Born weighting of its noise does not give Born statistics is known \cite{pearle89,wiseman96}, and the two-record recurrence is stated in Refs.~\cite{gtdborn,projectvii}; the exact solution adds the multi-record identity, the erosion law, and the identification of the entire difference from the CSL law as a drift tangent to the sphere of normalised states, to which norm conservation is blind.
\item \emph{The Born condition in history form (Proposition~\ref{prop:history}).} For the linear law driven by a c-number history with diffusion matrix $\mathsf D$, counter-term coefficient $\varkappa$ and conditional history drift $\mathsf b$, the drift of every record weight is computed in closed form. All weights are martingales only if $\mathsf b=2\mathsf D\bar\ell$ up to terms that do not distinguish records, so the history must lean towards the leading record with gain twice its own diffusion; and the counter-term must be present with unit coefficient, $\varkappa=1$, exactly when the squared eigenvalue vector is not an affine function on the records---forced by three collinear eigenvalues, as for a collective operator, invisible for two symmetric records and for the local detector channels of the Solvay model, and of the wrong sign in the white-noise limit of a smooth drive. The covariant family $\varkappa=2\alpha-1$, $\mathsf b=2\alpha\mathsf D\bar\ell$ interpolates between the bare law and CSL and reduces for two records to the fractional-response family of Ref.~\cite{projectvii}; for $\alpha\neq1$ the ensemble evolution is not affine in the density matrix, so it signals.
\item \emph{An Einstein relation for the sea (Proposition~\ref{prop:einstein}).} If the history is the integrated coefficient of a classical sea in canonical equilibrium at temperature $T$, coupled with strength $g$ and responding to a record expectation with strength $\lambda_{\rm b}$, the classical fluctuation--dissipation relation gives the gain $\alpha=\lambda_{\rm b}/(4gT\tau_c)$, $\tau_c$ the coefficient's correlation time. The Born gain is $\lambda_{\rm b}=4gT\tau_c$; causal consistency of the records fixes this ratio; the counter-term is a separate demand, which a smooth coefficient drive does not meet; the Markov law is the limit $\tau_c\ll$ selection time and the emergent law has memory otherwise.
\item \emph{Where state dependence may enter a collapse rate (Propositions~\ref{prop:ent} and \ref{prop:bound}).} A rate evaluated on the entanglement of the conditioned joint state lets an unread remote measurement change local statistics when it multiplies the ensemble generator. It may instead enter through the unravelling of a fixed Lindblad equation, the freedom identified by Bassi, D\"urr and Hinrichs \cite{uniqueness}, who note that it ranges from the random-phase law to the CSL law. This does not signal at the level of the density matrix, provided the fluctuation history is not itself a local record; but it can only suppress the realised selection rate below the value fixed by the operators, and it is invisible in every ensemble prediction. Enhancement beyond that value must be carried by the spectrum of fixed collective operators, as the $N^2$ amplification of CSL is.
\item \emph{The indefinite length (Section~\ref{sec:geodesicnorm}).} If the geodesic length proposed as the origin of norm conservation is built from the split-signature form of the aikyon algebra, a conserved positive value exists on the positive cone, but a probability norm on a linear space of states requires a positive-definite form, which exists only on a subspace of at most half the dimension. Selecting such a subspace is the same kind of choice as selecting a Lorentzian sector from the six-dimensional form.
\item \emph{Positive states do not see the Grassmann coefficient (Lemma~\ref{lem:positive}).} On the exterior algebra of the dynamical Grassmann generators every positive normalised functional annihilates the bifermionic coefficient, by the Cauchy--Schwarz argument that sends $\beta_1\beta_2$ to zero and by the nilpotent-ideal theorem of Ref.~\cite{physrecon}. The block whose kinetic form is the indefinite coefficient $J$ admits no state in which $J$ is visible; on the Clifford reconstruction the kinetic form is positive semidefinite in every state and definite unless the state is a pure sector; the real kinetic energy of the block remains an input, since the auxiliary factor has imaginary expectation in every positive state. The question whether a sea must share the $(2,2)$ split of the Grassmann block is thereby replaced by the question of which reconstruction is the physical one.
\item \emph{Coefficient noise from a continuum, and from the programme's own sea (Proposition~\ref{prop:bath}, Example~\ref{ex:sutherland}).} The free bifermionic coefficient is a constant of motion. Once its velocities form a stationary Gaussian process it is a mixing process whose white-noise strength follows in closed form from the velocity spectra, $D_J=4\sigma_1^2\sigma_2^2/\Gamma$ in the Ohmic case, and a finite bath recurs instead. The sea is then taken from the relational master-matrix extension of GTD \cite{master,projectvii}: the Sutherland texture is a chain with modes $\omega_\ell\propto\ell(M-\ell)$, which in the thermodynamic limit at fixed spacing---an inserted choice, since the texture is written on a fixed circle---is a gapless band with linear dispersion, an exact period $\propto M^2$ and low-frequency mode spacing $4\Omega_{\max}/M$. A reconstructed block coupled to its local strain acquires an Ohmic spectral density with friction $\gamma=\pi^3c^2/(16mm_b\Omega_{\max}^3)$, continuum criterion $4\Omega_{\max}/M\ll\gamma$ and zero-frequency line $T/M_{\rm eff}$, $M_{\rm eff}=m_b[1+\pi\gamma M/(12\Omega_{\max})]$, confirmed by the exact finite-$M$ dynamics at $M=4096$. A purely kinetic coupling, the structure of the cross terms of the GTD action, conserves the block momentum and gives no friction. The coupling to the strain is inserted, not derived.
\end{enumerate}
Around these the paper keeps, briefly, the conditional results of the effective analysis---a Gaussian instrument law selects one Solvay record asymptotically with Born frequencies, and its records, including those of an EPR--Bell pair, are non-signalling whenever the emergent operations are local and complete---and the negative microscopic results that the free coefficient is frozen and that an unchanged positive representation of the auxiliary algebra sends it to zero, so that a nonhomomorphic physical projection is required. Section~\ref{sec:closure} states, in one equation, the calculation on which a microscopic account of collapse in GTD now depends. Where the present paper overlaps with the author's recent audits of the programme \cite{projectvii,physrecon,audit,microclosure}, which establish that the additive GTD action carries no inter-atom exchange current, that the free master matrix recurs at every size, and that no written action contains an irreversible sector---while leaving open an interacting thermodynamic limit---the overlap is stated at the point of use; those papers specify requirements, and the present one tests two of them in a model built from the programme's own interacting sector.

\subsection{Outline and conventions}
Section~\ref{sec:plain} has just given the argument without equations. Section~\ref{sec:setting} describes the pregeometric setting. Section~\ref{sec:joint} builds the effective collapse law and contains Propositions~\ref{prop:bare} and \ref{prop:history} and the discussion of the geodesic length. Section~\ref{sec:causal} applies the law to Solvay and Bell records, states the no-signalling theorem, and contains Propositions~\ref{prop:ent} and \ref{prop:bound}. Section~\ref{sec:diagnostics} follows the GTD coefficient from the free theory through the linearised bath, Proposition~\ref{prop:bath}, the Sutherland sea of Example~\ref{ex:sutherland}, Lemma~\ref{lem:positive} on the block that carries it, and the Einstein relation of Proposition~\ref{prop:einstein}, to the projection. Section~\ref{sec:closure} states what remains and collects the answers in Table~\ref{tab:status}. Appendices record the microscopic action and its projection tests, the fermionic reservoir, the probability details of the joint law, and the linearised-bath model with the Sutherland sea.

Stochastic equations are in It\^o form unless a Stratonovich circle $\circ$ appears. Wiener processes $W_a$ have independent increments with $\dd W_a\dd W_b=\delta_{ab}\dd\tau$ under the physical law. The operators $L_a$ that distinguish records have units of (time)$^{-1/2}$; record weights and exposures are dimensionless. Ancillary scripts reproduce every numerical statement.

\section{The pregeometric setting}\label{sec:setting}
In an ordinary measurement calculation the positions of the detectors and the readings of their clocks are given at the outset. In GTD these too must be explained: the matter that serves as ruler or clock acquires approximately classical behaviour through localisation, so the causal structure relevant to a laboratory belongs to a later description than the fundamental variables. This section records what geometric information those variables already carry, and in what sense collapse can happen ``before spacetime''.

\subsection{The split-signature form and its null cone}\label{sec:splitform}
The proposal associates six candidate spacetime directions with the split-biquaternionic sector of the aikyon algebra \cite{wsi,book}. Let
\begin{equation}
\mathcal B=\mathbb H\otimes\mathbb D,\qquad
\mathbb D=\mathbb R[\omega]/(\omega^2-1),\qquad
\mathcal B\simeq\mathbb H\oplus\mathbb H,
\label{eq:split_algebra}
\end{equation}
where $\mathbb H$ is the real quaternion algebra and $\omega$ is a commuting unit with $\omega^2=+1$; the idempotents $(1\pm\omega)/2$ split $\mathcal B$ into the two quaternion copies. The six directions of interest are the imaginary quaternion directions in the two copies,
\begin{equation}
V=\operatorname{Im}\mathbb H\oplus\omega\operatorname{Im}\mathbb H,
\qquad v=\mathsf u+\omega\mathsf x,\quad \mathsf u=\sum_{r=1}^3u_r e_r,\quad \mathsf x=\sum_{r=1}^3x_r e_r,
\end{equation}
with $e_re_s=-\delta_{rs}+\epsilon_{rst}e_t$. This is a vector subspace, not a subalgebra. With the conjugation $(a+\omega b)^{\sim}=\bar a-\omega\bar b$ (bar denoting quaternion conjugation) one has $\tilde v=-\mathsf u+\omega\mathsf x$ and
\begin{equation}
v\tilde v=(|u|^2-|x|^2)1+\omega[\mathsf u,\mathsf x],\qquad
\tfrac12(v\tilde v+\tilde v v)=(|u|^2-|x|^2)1.\label{eq:sym_norm}
\end{equation}
The symmetrisation matters, since $v\tilde v$ alone has a non-scalar remainder. With the opposite overall sign the real quadratic form is
\begin{equation}
\mathsf q(v)=-u_1^2-u_2^2-u_3^2+x_1^2+x_2^2+x_3^2,
\label{eq:splitnorm}
\end{equation}
of signature $(3,3)$. Its status is the following. The split signature comes from the choice of conjugation: the conjugation $(a+\omega b)^{\wedge}=\bar a+\omega\bar b$ gives $|u|^2+|x|^2$ instead, positive definite. So \eqref{eq:splitnorm} is the form that carries the spacetime interpretation of Ref.~\cite{wsi}, and that interpretation is what singles it out. And $\mathsf q$ is invariant under $SO(3,3)$ acting on the six coefficients, which says nothing about whether those transformations preserve the algebra multiplication or a microscopic action.

The zero set of $\mathsf q$ is an algebraic null cone, and the form motivates the $SO(3,3)$ BF construction of Ref.~\cite{wsi}, in which a six-dimensional geometry is related to two overlapping four-dimensional Lorentzian sectors. A null cone, however, is not yet a causal order. In Lorentzian geometry the timelike vectors at a point fall into two disconnected components, and a time orientation names one of them the future. With three negative directions a vector can turn into its negative without ever becoming null:
\begin{equation}
v(\vartheta)=(\cos\vartheta,\sin\vartheta,0;0,0,0),\qquad \mathsf q(v(\vartheta))=-1,\qquad v(\pi)=-v(0),
\label{eq:rotation}
\end{equation}
and the rotation lies in the identity component of $SO(3,3)$. The negative-norm set is connected, as is the positive-norm set. So the quadratic form supplies a null cone but not the two-component timelike structure from which a future--past order is built. Restricting to $u_2=u_3=0$ leaves $-u_1^2+|x|^2$, which has the ordinary cones; this shows what a selected sector could provide, not why it is selected. Craig and Weinstein have shown that ultrahyperbolic equations admit a well-posed Cauchy problem only with a nonlocal constraint on codimension-one data \cite{craigweinstein}, so a propagation law on the full six-dimensional form would in any case have to be stated explicitly rather than read off from the signature. None of this rules out additional causal data; algebraic formulations of causality in noncommutative geometry exist with extra operator structure \cite{franco}.

\subsection{Connes time and the hypersurface question}\label{sec:connes}
We write $\tau$ for the ordering parameter of the microscopic dynamics. Its identification with the modular flow of a represented algebra and state is part of the GTD proposal \cite{connes,gtdtime}; for the present paper $\tau$ is simply the parameter in which the fundamental equations, and the stochastic law that coarse-graining is supposed to produce, are first order. Two consequences follow, one helpful and one demanding.

The helpful one is the point made in the Introduction. A collapse law written in $\tau$ is not a law on spacetime and is not asked to be covariant under spacetime transformations; the hypersurface on which collapse ``happens'' is not a question that arises at this level. Relativistic collapse models formulated on spacetime have to answer it explicitly, and they do so by adding structure: nonstandard degrees of freedom, a flash ontology, or a foliation-dependent construction \cite{aharonov81,ggp,tumulka,bedingham}. GTD proposes a different route: selection is placed before the emergence of the arena in which the question would be posed, and the question returns as a constraint on what emerges.

The demanding one is that a dynamics which is first order and Markovian in a single parameter $\tau$ defines a preferred ordering of events. Whether the emergent spacetime description inherits that ordering as a preferred foliation, and whether its records are nevertheless Lorentz covariant to the required accuracy, is a prediction that the completion must make. It is not something the present paper can grant, and we return to it in Section~\ref{sec:continuum}. Modular flows, moreover, are automorphism flows: a Connes time is not by itself a selection mechanism. The selection mechanism has to come from the coarse-graining of the interacting dynamics, which is the subject of Section~\ref{sec:diagnostics}.

A full reconstruction would relate collective microscopic configurations to classical events, select a Lorentzian sector and a time orientation, establish the propagation laws of its matter, and show when selected configurations behave as bodies and clocks. The gravi--weak BF analysis \cite{wsi} and the spectral-action analysis \cite{spectral} supply candidate geometric and low-energy descriptions, but both are formulated on an already six-dimensional manifold, and an action for a selected sector does not determine the probabilities with which that sector is selected. The effective analysis of Sections~\ref{sec:joint}--\ref{sec:causal} assumes these tasks have been carried out to the point where approximately classical detectors, a laboratory spacetime and local operations on it are available, and tests whether a proposed collapse law is consistent with that description.

\subsection{Fluctuations in trace dynamics and the anti-self-adjoint term}\label{sec:asa}
In trace dynamics the coordinates $q_r$ and momenta $p_r$ are matrices with an action built from a trace, and global unitary invariance gives a conserved matrix charge $\widetilde C=\sum_{r\in B}[q_r,p_r]-\sum_{r\in F}\{q_r,p_r\}$, with $B$ and $F$ the bosonic and fermionic variables \cite{adler}. The equilibrium statistical mechanics of this matrix phase space, with $\widetilde C$ among the conserved charges, is Adler's route to quantum theory, and his Chapter~6 adds Brownian corrections to that equilibrium as the route to stochastic reduction, through energy-driven equations of the Hughston type \cite{adler,adler2024,hughston}. GTD extends the programme to geometry \cite{gtdcritical,book}. The present work belongs to Adler's programme: it uses stochastic reduction and Born martingales, and the distinction it draws is between a model of a fluctuation correction and the evaluation of that correction for a specified action and preparation.

The mechanism proposed for collapse is that a strongly coupled collective system interacts with unresolved aikyons, and that coarse-graining their anti-self-adjoint (ASA) contribution produces effective amplitude fluctuations. Why anti-self-adjoint? A self-adjoint Hamiltonian changes phases and mixes amplitudes unitarily; an anti-self-adjoint part can change their relative magnitudes. For an unnormalised vector $\phi$ and a candidate effective generator $H_S+i\hbar A$ with $H_S^\dagger=H_S$, $A^\dagger=A$,
\begin{equation}
\dot\phi=(-iH_S/\hbar+A)\phi,\qquad
\frac{\dd}{\dd\tau}\|\phi\|^2=2\langle\phi,A\phi\rangle,\qquad
\dot\psi=\left[-iH_S/\hbar+A-\langle A\rangle_\psi\right]\psi
\label{eq:detnorm}
\end{equation}
for the normalised $\psi=\phi/\|\phi\|$. The subtraction keeps the total squared amplitude fixed while its distribution among alternatives changes. It does not make the equation stochastic, and even when $A$ fluctuates one must still say how often its histories occur and what systematic response accompanies them. Charge conservation determines neither the number of effective noise channels nor their correlations. The GTD Born-rule analysis identifies a conditional noise--drift relation \cite{gtdborn}, the interaction analysis supplies a candidate coupling \cite{cyclic}, and the monograph distinguishes its real equilibrium Hamiltonian from the additional projection needed for an ASA collapse generator \cite[Secs.~7.3.6 and 9.2--9.4]{book}. We preserve that distinction throughout.

Mass-proportional CSL \cite{grw,csl,pearlesquires,bassi} uses a smeared mass-density operator on a background spacetime; our effective variables distinguish collective alternatives that need not initially denote positions, which lets the selection problem be posed before a metric is assigned but does not by itself explain spatial localisation. The probability construction below is drawn from established collapse and continuous-measurement theory \cite{pearle76,gisin84,wiseman,bouten,diosi88}; it adds no new solution to the measurement problem, and its purpose is to state an exact target for the GTD projection.

\section{The effective collapse law}\label{sec:joint}
We now assume a finite effective Hilbert space with a positive inner product and a set of stable alternatives, and build the simplest law that assigns, together, a change of state and the frequency with which that change occurs. During the selection calculation the alternatives are not mixed by the Hamiltonian, which isolates the reduction of an already formed superposition.

\subsection{One step: an update and its frequency}
Choose commuting self-adjoint operators $L_a$, $a=1,\ldots,m$, with joint eigenvectors $|j\rangle$, projectors $P_j$, and eigenvalue vectors $\ell_j=(\ell_{1j},\ldots,\ell_{mj})$. Think of $\ell_j$ as the list of numbers by which the fluctuations tell alternative $j$ apart; two alternatives with the same list cannot be distinguished by this mechanism. The operators stay fixed. All state dependence will come from the evolving weights, never from changing the apparatus according to a distant observer's conditioned state.

Let $q_h(\dd y)$ be the Gaussian measure of mean zero and covariance $h\id$ on $\mathbb R^m$ and define
\begin{equation}
M_h(y)=\exp\Bigl(\sum_aL_ay_a-h\sum_aL_a^2\Bigr),\qquad
\int M_h(y)^\dagger M_h(y)\,q_h(\dd y)=\id,
\label{eq:M}
\end{equation}
where $y$ is an $m$-component increment and $h$ a step in $\tau$; the completeness identity follows from $\int e^{2\ell_j\cdot y-2|\ell_j|^2h}q_h(\dd y)=1$. The physical operation is assumed to be the instrument
\begin{equation}
\mathcal I_{\dd y}(\rho)=M_h(y)\rho M_h(y)^\dagger\,q_h(\dd y),
\label{eq:instrument}
\end{equation}
which specifies both the probability $\Tr\mathcal I_{\dd y}(\rho)$ of an increment in $\dd y$ and the conditional state after it. Completeness makes the probabilities sum to one, also for a system entangled with something else. This operational trace rule is a premise: it already contains the Born prescription, and nothing below derives that prescription from calling $y$ noise. For record weights $p_j$, completing the square gives the exact joint law
\begin{align}
\mathbb P(\dd y\mid p)&=\sum_jp_j\,\mathcal N(2\ell_jh,h\id)(\dd y),\label{eq:mix}\\
p'_j&=\frac{p_j\exp(2\ell_j\cdot y-2|\ell_j|^2h)}
{\sum_kp_k\exp(2\ell_k\cdot y-2|\ell_k|^2h)}.\label{eq:update}
\end{align}
The first line says how often an increment occurs, the second what it does to the weights; the second is Bayes' rule with the increment as data. Multiplying \eqref{eq:update} by the mixture density in \eqref{eq:mix} cancels the denominator and leaves $p_j$ times a normalised Gaussian, so
\begin{equation}
\E[p'_j\mid p]=p_j,\qquad \E[y\mid p]=2\bar\ell h,\qquad
\operatorname{Cov}(y\mid p)=h\id+4h^2\Bigl(\sum_jp_j\ell_j\ell_j^T-\bar\ell\bar\ell^T\Bigr),
\qquad\bar\ell=\sum_jp_j\ell_j.\label{eq:moments}
\end{equation}
Figure~\ref{fig:joint} shows the two parts of one step. Replacing the physical increment distribution by the reference Gaussian while keeping the update destroys the first identity: the mean weight is preserved only because the histories are distributed as \eqref{eq:mix}. This is the distinction between the ``ostensible'' and the actual distribution of the noise \cite{wiseman96}, which Section~\ref{sec:normdrift} makes quantitative.

\begin{figure}[!htbp]
\centering
\includegraphics[width=\textwidth]{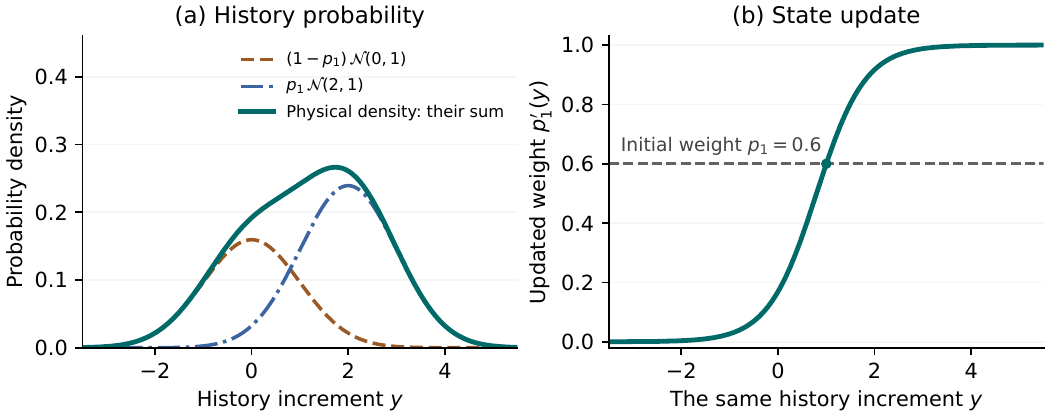}
\caption{One Gaussian step has two parts. In rescaled units $h=1$, $\ell_0=0$, $\ell_1=1$, $p_1=0.6$. (a)~The physical density of the history increment, with its two weighted components. (b)~How that same increment changes the record weight. The posterior varies from run to run, but its average over the density in (a) is exactly $0.6$. Changing the history density while keeping the update changes that average.}\label{fig:joint}
\end{figure}

\subsection{Continuous limit, martingale weights, and selection}\label{sec:continuous}
Repeated steps have a diffusion limit. With $\bar\ell_a=\langle L_a\rangle$ and $W_a$ the innovation Wiener processes under the physical law,
\begin{align}
\dd Y_a&=2\bar\ell_a\dd\tau+\dd W_a,\label{eq:Y}\\
\dd p_j&=\sum_ab_{ja}(p)\dd W_a,\qquad b_{ja}(p)=2p_j(\ell_{aj}-\bar\ell_a),\label{eq:p}\\
\dd\psi&=\Bigl[-\frac{iH}{\hbar}-\frac12\sum_a(L_a-\langle L_a\rangle)^2\Bigr]\psi\dd\tau
+\sum_a(L_a-\langle L_a\rangle)\psi\dd W_a,\label{eq:SSE}
\end{align}
with $H$ diagonal in the record basis, or zero, for the weight equation and the selection proof. Equation~\eqref{eq:SSE} is the CSL-type stochastic Schr\"odinger equation \cite{csl,gisin84,diosi88}; its nonlinear subtractions keep the norm constant, and \eqref{eq:p} has no drift, so each weight is a bounded martingale. The history $Y$ carries a conditional signal $2\bar\ell$ and $W$ is its unpredictable remainder.

The standard way to make the probability premise explicit \cite{wiseman,bouten} is a reference measure $\mathbb Q$ under which the $Y_a$ are independent Wiener processes. The linear equation $\dd\phi=K\phi\dd\tau+\sum_aL_a\phi\dd Y_a$ has a norm with zero reference drift for every state exactly when $K+K^\dagger+\sum_aL_a^\dagger L_a=0$; with $K=-iH/\hbar-\frac12\sum_aL_a^2$ and
\begin{equation}
\dd\mathbb P|_{\mathcal F_\tau}=\|\phi_\tau\|^2\,\dd\mathbb Q|_{\mathcal F_\tau},\label{eq:measure}
\end{equation}
normalisation and the change of drift give \eqref{eq:SSE}. Appendix~\ref{app:joint} records the joint generator and the consistency of the finite-time laws. A microscopic derivation cannot compute some other law and then impose \eqref{eq:measure} by hand; the weight in \eqref{eq:measure} is the content of the Born premise.

Preserving the mean weights is necessary but not sufficient: a state that never moved would preserve them too. For constant $L_a$ the history law has the latent representation ``draw a record $j^\star$ with weights $p_j(0)$, then $Y_\tau=2\ell_{j^\star}\tau+\widetilde W_\tau$'' with $\widetilde W$ a Wiener process independent of $j^\star$; this is the exact Girsanov solution familiar from energy-driven reduction \cite{hughston,abbh}, and $\widetilde W$ is not the innovation $W$. Substituting in \eqref{eq:update},
\begin{equation}
\log\frac{p_j(\tau)}{p_{j^\star}(\tau)}
=\log\frac{p_j(0)}{p_{j^\star}(0)}
+2(\ell_j-\ell_{j^\star})\cdot\widetilde W_\tau-2|\ell_j-\ell_{j^\star}|^2\tau.
\label{eq:selection}
\end{equation}
Since $\widetilde W_\tau/\tau\to0$ almost surely, every weight whose signature differs from $\ell_{j^\star}$ vanishes: one of finitely many distinct eigenvalue vectors is selected asymptotically, with probability equal to its initial weight. Degenerate signatures stay unresolved. The latent $j^\star$ is a representation of the assumed history law, not a hidden GTD record; the result shows that the Born premise is consistent with dynamical selection, not that it is derived. We call \eqref{eq:SSE} the CSL law below.

\subsection{Norm conservation does not fix the drift: the bare law solved exactly}\label{sec:normdrift}
Suppose a positive normalised state has emerged, by whatever mechanism, and ask what conserving its norm requires of its stochastic dynamics. For an It\^o equation $\dd\psi=f(\psi)\dd\tau+\sum_aB_a(\psi)\dd W_a$ on the unit sphere, the product rule gives the necessary and sufficient local conditions
\begin{equation}
\operatorname{Re}\langle\psi,B_a\rangle=0,\qquad
2\operatorname{Re}\langle\psi,f\rangle+\sum_a\|B_a\|^2=0.
\label{eq:norm_conditions}
\end{equation}
The noise must start out tangent, and the drift must have an inward component to compensate the order-$\dd\tau$ contribution of squared increments. Neither condition says anything about a drift tangent to the sphere, and a tangent drift is what moves the state between alternatives.

The simplest norm-preserving law one can write is the one a coarse-graining produces if it merely centres a fluctuating gain: let $A$ in \eqref{eq:detnorm} be $\sum_aL_a\xi_a(\tau)$ with smooth noises $\xi_a$ of short correlation time, so that $\dot\psi=\sum_a\xi_aC_a\psi$ with $C_a=L_a-\langle L_a\rangle$, and take the white-noise limit, which by the Wong--Zakai theorem is the Stratonovich equation (the L\'evy-area correction of the multidimensional limit vanishes because the noise vector fields commute for commuting $L_a$)
\begin{equation}
\dd\psi=\sum_aC_a\psi\circ\dd W_a .
\label{eq:bare_tangent_noise}
\end{equation}
It preserves the norm, since $\langle\psi,C_a\psi\rangle=0$ and the ordinary chain rule holds. We call it the bare law. It is the normalisation of the linear law $\dd\phi=\sum_aL_a\phi\circ\dd W_a$, whose solution for commuting $L_a$ is $\phi_\tau=e^{\sum_aL_aW_a(\tau)}\phi_0$: centred noise exponentiated, with neither the counter-term $-\frac12\sum_aL_a^2$ of the CSL update nor the reweighting \eqref{eq:measure} of its noise. Its exact solution shows what the absence of those two ingredients costs.

\begin{proposition}[The bare law never selects]\label{prop:bare}
Let the $L_a$ commute, with joint eigenvectors $|j\rangle$ and eigenvalue vectors $\ell_j$, and let $p_j=|\langle j|\psi\rangle|^2$ with all $p_j(0)>0$. Under \eqref{eq:bare_tangent_noise}, for every pair of records,
\begin{equation}
\log\frac{p_j(\tau)}{p_k(\tau)}=\log\frac{p_j(0)}{p_k(0)}+2(\ell_j-\ell_k)\cdot W_\tau
\qquad\text{exactly.}
\label{eq:bare_exact}
\end{equation}
Consequently:
\begin{enumerate}
\item[(i)] each log-ratio with $\ell_j\neq\ell_k$ is a driftless one-dimensional Brownian motion, hence recurrent: it returns to zero infinitely often, so no pair of distinguishable records is ever permanently resolved;
\item[(ii)] for two records with $\delta=\ell_1-\ell_0$ and $p=p_1$, writing $\logit p:=\log[p/(1-p)]$ for the log-odds, $\;\logit p(\tau)=\logit p_0+2\delta\cdot W_\tau$; the difference between $p(\tau)$ and the indicator of $\{\delta\cdot W_\tau>0\}$ tends to zero in probability, so $p(\tau)$ converges in law to a fair Bernoulli variable for every $p_0\in(0,1)$, and $\E p(\tau)-\tfrac12\simeq\logit p_0/(2|\delta|\sqrt{2\pi\tau})\to0$;
\item[(iii)] the It\^o form of the two-record law is
$\dd p=2|\delta|^2p(1-p)(1-2p)\dd\tau+2p(1-p)\,\delta\cdot\dd W$.
\end{enumerate}
\end{proposition}
\begin{proof}
In the record basis \eqref{eq:bare_tangent_noise} reads $\dd c_j=\sum_a(\ell_{aj}-\bar\ell_a)c_j\circ\dd W_a$, with real coefficients, so the phases of the $c_j$ are constant and the Stratonovich chain rule gives $\dd\log|c_j|=\sum_a(\ell_{aj}-\bar\ell_a)\circ\dd W_a$. In the difference of two such equations the state-dependent $\bar\ell_a$ cancels and the coefficient is constant, so the Stratonovich and It\^o integrals coincide; doubling gives \eqref{eq:bare_exact}. (i) A fixed nonzero linear combination of independent Wiener processes is a Brownian motion, which in one dimension is recurrent \cite{mortersperes}. (ii) The Gaussian $2\delta\cdot W_\tau$ has variance $4|\delta|^2\tau\to\infty$, and the logistic function differs from the step function by an odd integrable function; the expansion of $\E\,\logit^{-1}(x_0+sZ)$ for $Z$ standard normal and large $s$ gives $\tfrac12+x_0/(s\sqrt{2\pi})+O(s^{-3})$. (iii) It\^o's formula for $p=\logit^{-1}(g)$ with $\dd g=2\delta\cdot\dd W$.
\end{proof}

\begin{figure}[!htbp]
\centering
\includegraphics[width=\textwidth]{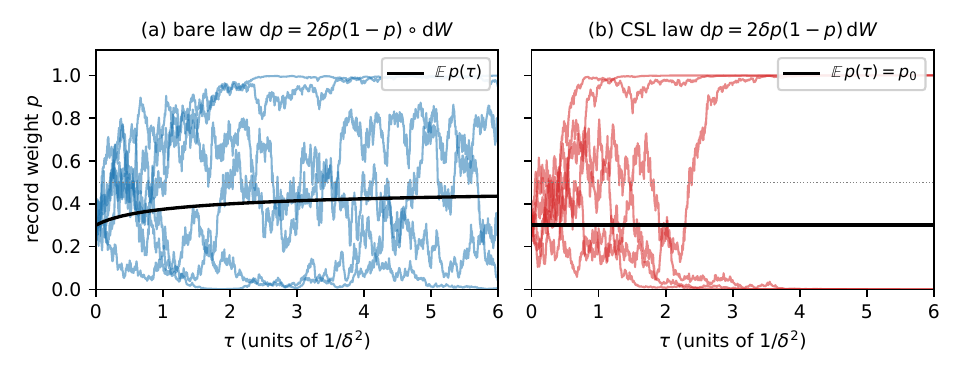}
\caption{Six sample paths of the record weight from $p_0=0.3$ with $\delta=1$; the circle in (a) denotes a Stratonovich differential, and the absence of a drift term in (b) is in the It\^o sense. (a)~The bare law \eqref{eq:bare_tangent_noise}, drawn from its exact solution \eqref{eq:bare_exact} for the log-odds $\logit p=\log[p/(1-p)]$: paths wander between the two records and recross $p=1/2$ (dotted); the mean weight (black) climbs from $0.3$ towards $1/2$, reaching $0.435$ at $\tau=6$. (b)~The CSL law \eqref{eq:SSE}: each path settles on one record, in the ensemble a fraction $p_0$ of them on record~1, and the mean weight is exactly $p_0$. The two laws have the same noise field and the same conserved norm and differ only by the tangent drift \eqref{eq:tangent_response}; the paths in (b) use independent noise realisations.}\label{fig:bare}
\end{figure}

Figure~\ref{fig:bare} shows the contrast. Under the bare law the late-time state is usually near one record or the other, because $|W_\tau|$ grows, but which one keeps changing, and the two are visited with equal frequency whatever the initial weights; the fraction of a long run spent near either record follows L\'evy's arcsine law rather than converging to a constant \cite{mortersperes}. Under the CSL law, by \eqref{eq:selection}, the log-ratio of any record against the selected one is a Brownian motion with drift $-2|\ell_j-\ell_{j^\star}|^2$, transient, and the selected record is drawn with Born frequencies. In the two-record case the two laws differ precisely by a drift in the log-odds,
\begin{equation}
\dd\logit p=2\delta\cdot\dd W\ \ (\text{bare}),\qquad
\dd\logit p=2\delta\cdot\dd W+2|\delta|^2(2p-1)\dd\tau\ \ (\text{CSL}),
\label{eq:logit_compare}
\end{equation}
equivalently by the It\^o drift $-2|\delta|^2p(1-p)(1-2p)$ in $p$, which at $p=0.3$ is $-0.168\,|\delta|^2$. That drift is what a coarse-graining must supply in addition to the centred noise. That a norm-preserving law whose noise carries no Born weighting fails to reproduce Born statistics is known \cite{pearle89,wiseman96}; the exact solution adds that it does not select either.

In state-space language the missing term is a tangent vector field. For one channel, write $\ell_\psi=\langle\psi,L\psi\rangle$, $C_\psi=L-\ell_\psi$ and $\Delta_\psi^2=\langle\psi,C_\psi^2\psi\rangle$ for the variance. Converting \eqref{eq:bare_tangent_noise} to It\^o form requires the derivative of the noise field along itself; since the directional derivative of $\ell_\psi$ along $C_\psi\psi$ is $2\Delta_\psi^2$,
\begin{equation}
D(C_\psi\psi)[C_\psi\psi]=C_\psi^2\psi-2\Delta_\psi^2\psi,\qquad
f_0(\psi)=\tfrac12C_\psi^2\psi-\Delta_\psi^2\psi,
\label{eq:bare_ito_drift}
\end{equation}
whereas the CSL drift in \eqref{eq:SSE} is $f_{\rm CSL}=-\tfrac12C_\psi^2\psi$. Their difference,
\begin{equation}
R(\psi):=f_{\rm CSL}-f_0=\bigl(\Delta_\psi^2\,\id-C_\psi^2\bigr)\psi,\qquad\langle\psi,R(\psi)\rangle=0,
\label{eq:tangent_response}
\end{equation}
is tangent to the sphere; for several commuting channels $R=\sum_a(\Delta_{a,\psi}^2-C_a^2)\psi$, without cross terms. Both $f_0$ and $f_{\rm CSL}$ satisfy \eqref{eq:norm_conditions}. Norm conservation therefore cannot distinguish a law that never selects from one that selects with Born frequencies; the distinction lives in a response that conserving the norm does not see. This is the state-space form of the noise--drift condition of Ref.~\cite{gtdborn}: if a derived two-record process has the Stratonovich form $\dd p=a(p)\dd\tau+g(p)\circ\dd W$, mean preservation requires $a=-\frac12gg'$, and for $g=2|\delta|p(1-p)$ this is $-2|\delta|^2p(1-p)(1-2p)$, which is what \eqref{eq:tangent_response} supplies. The condition is a test for a derived law, not a notational convention; a coarse-graining can produce it, fail it, or produce memory instead.

\begin{remark}[The finite-time version]\label{rem:gain}
The same failure appears whenever a history $y$ acts over a fixed time through $p$-independent diagonal multipliers $M_1(y)$, $M_0(y)$ with a $p$-independent law $\nu$, so that $p_y=pr/(1-p+pr)$ with $r=|M_1|^2/|M_0|^2$. If $\E_\nu p_y=p$ for every $p$, then $r=1$ almost surely: with $F(p)=\E[pr/(1-p+pr)]$, differentiation under the expectation (the integrand's $p$-derivative is bounded by $1/[4p(1-p)]$) gives $F'(\frac12)=\E\,4r/(1+r)^2=1-\E\bigl((r-1)/(r+1)\bigr)^2$, and $F(p)=p$ forces the last term to vanish. For instance $r\in\{3,\frac13\}$ with equal probability sends $p=0.8$ to $\E p_y=0.747$. Proposition~\ref{prop:bare} is the continuous-time case, $r=e^{2\delta\cdot W_\tau}$. Two normalised updates are legitimate states; what fails is their frequency assignment, which is why normalisation alone cannot settle the Born problem. That a static distribution over a state-independent gain cannot reproduce the Born rule is stated in Ref.~\cite{gtdborn}; the computation above is its finite-time proof. Ref.~\cite{kakade} reaches the CSL form from a deterministic nonunitary law by adopting an It\^o noise ansatz at the coarse-grained level; the present result shows that the ansatz is doing the work, since without it the natural limit is the bare law. Ref.~\cite{projectvii} strengthens Proposition~\ref{prop:bare} in the two-record case: with a fraction $\alpha$ of the required drift the process still completes for every $\alpha>0$, with exit law $I_{p_0}(\alpha,\alpha)$, and it is Born only at $\alpha=1$. Section~\ref{sec:history} identifies $\alpha$ with a response gain.
\end{remark}

A further comparison will be needed in Section~\ref{sec:failures}. For the same $L$, the random-phase law $\dd\psi=-iL\psi\dd W-\frac12L^2\psi\dd\tau$, solved by $\psi_\tau=e^{-iLW_\tau}\psi_0$, preserves the norm and every weight and selects nothing, yet it has the same ensemble equation $\dot\rho=-\frac12[L,[L,\rho]]$ as the CSL law \cite{wiseman,bouten}. Decay of averaged interference cannot tell these histories apart. A microscopic derivation has to determine the realised histories, not only their average; Section~\ref{sec:failures} shows that this freedom is also where state dependence can hide.

\subsection{The same condition in history form: counter-term and response gain}\label{sec:history}
A microscopic derivation does not produce $\psi$ directly. It produces a c-number history---in GTD, the integrated coefficient of Section~\ref{sec:diagnostics}---and a linear law for the unnormalised state driven by that history; the normalised state and the record weights follow. It is therefore useful to state the Born condition entirely in terms of what such a derivation has to deliver: the deterministic part of the linear law and the conditional law of the history. Equation~\eqref{eq:Y} already contains the answer for the CSL case, $\dd Y_a=2\bar\ell_a\dd\tau+\dd W_a$. The following proposition separates the two ingredients and quantifies what partial compliance costs.

\begin{proposition}[The Born condition in history form]\label{prop:history}
Let the $L_a$ commute with eigenvalue vectors $\ell_j$ on the records $|j\rangle$, and let the unnormalised state obey the linear law
\begin{equation}
\dd\phi=\Bigl[-\frac\varkappa2\sum_{ab}\mathsf D_{ab}L_aL_b\,\dd\tau+\sum_aL_a\dd Y_a\Bigr]\phi,
\label{eq:linear_history}
\end{equation}
where the history $Y$ is a continuous semimartingale with quadratic covariation $\dd\langle Y_a,Y_b\rangle=\mathsf D_{ab}\dd\tau$, $\mathsf D$ constant and positive definite, and drift $\mathsf b_\tau\dd\tau$ with respect to the history's own filtration, $\mathsf b_\tau$ a function of the past of $\psi$. Let $\psi=\phi/\|\phi\|$, $p_j=|\langle j|\psi\rangle|^2$, $\bar\ell=\sum_kp_k\ell_k$ and $\langle\ell^T\mathsf D\ell\rangle=\sum_kp_k\ell_k^T\mathsf D\ell_k$. Then:
\begin{enumerate}
\item[(i)] the It\^o drift of each record weight is
\begin{equation}
\frac{\E[\dd p_j\mid\mathcal F_\tau]}{\dd\tau}=p_j\Bigl[2(\ell_j-\bar\ell)\cdot\bigl(\mathsf b_\tau-2\mathsf D\bar\ell\bigr)+(1-\varkappa)\bigl(\ell_j^T\mathsf D\ell_j-\langle\ell^T\mathsf D\ell\rangle\bigr)\Bigr];
\label{eq:history_drift}
\end{equation}
\item[(ii)] every weight is a martingale for every state if and only if $\mathsf b_\tau=2\mathsf D\bar\ell+c_\tau$ with $\ell\mapsto2\ell\cdot c_\tau+(1-\varkappa)\ell^T\mathsf D\ell$ constant on the records; components of $c_\tau$ orthogonal to all differences $\ell_j-\ell_k$ are immaterial. The counter-term coefficient is therefore forced to $\varkappa=1$ exactly when $\ell\mapsto\ell^T\mathsf D\ell$ is not an affine function on the set of records; three records with collinear eigenvalue vectors suffice, for instance one channel with three distinct eigenvalues and the other channels equal. It is never forced for two records, nor whenever $\ell_j^T\mathsf D\ell_j$ takes one value on all records, as for the local detector channels of Section~\ref{sec:solvay} with equal diffusion, $\ell_j=e_j$, and more generally whenever the records are affinely independent, since a constant $c$ then absorbs any $\varkappa$. The CSL law is $\varkappa=1$, $\mathsf b=2\mathsf D\bar\ell$; the bare law is $\varkappa=-1$, $\mathsf b=0$; and $\varkappa=-1$ is what the white-noise limit of any smooth real drive produces, with any gain;
\item[(iii)] the one-parameter family that interpolates between the bare law and the CSL law covariantly under $L_a\mapsto L_a+s_a\id$ is $\varkappa=2\alpha-1$, $\mathsf b=2\alpha\mathsf D\bar\ell$, with drift
\begin{equation}
\frac{\E[\dd p_j\mid\mathcal F_\tau]}{\dd\tau}=2(1-\alpha)\,p_j\Bigl[(\ell_j-\bar\ell)^T\mathsf D(\ell_j-\bar\ell)-\bigl\langle(\ell-\bar\ell)^T\mathsf D(\ell-\bar\ell)\bigr\rangle\Bigr];
\label{eq:fractional_drift}
\end{equation}
for two records with $\delta=\ell_1-\ell_0$ this is $2(1-\alpha)\,\delta^T\mathsf D\delta\;p(1-p)(1-2p)$, the fractional-response family of Ref.~\cite{projectvii}, which completes for every $\alpha>0$ with exit law $I_{p_0}(\alpha,\alpha)$, is recurrent for $\alpha\leq0$, and is Born only at $\alpha=1$; $\alpha=0$ is Proposition~\ref{prop:bare}. With $\varkappa$ held at $1$ and only the gain reduced to $2\alpha\mathsf D\bar\ell$, the drift $4(\alpha-1)p_j(\ell_j-\bar\ell)^T\mathsf D\bar\ell$ depends on the zero of the $L_a$ and agrees with \eqref{eq:fractional_drift} only when $\ell_j^T\mathsf D\ell_j$ is the same on all records, for two records when the eigenvalues are symmetric about zero;
\item[(iv)] for $\alpha\neq1$ the ensemble evolution is not affine in $\rho$: a mixture of record eigenstates is stationary, while a mixture of superpositions with the same weights has $\E p_j$ drifting by \eqref{eq:fractional_drift}; two preparations of the same density matrix therefore evolve differently, and the law signals by the construction of Section~\ref{sec:failures}.
\end{enumerate}
\end{proposition}
\begin{proof}
In the record basis \eqref{eq:linear_history} gives $\dd\phi_j=\phi_j[\ell_j\cdot\dd Y-\frac\varkappa2\ell_j^T\mathsf D\ell_j\dd\tau]$ with real coefficients, so the phase of $\phi_j$ is constant and $a_j:=\log|\phi_j|^2$ obeys $\dd a_j=2\ell_j\cdot\dd Y-(1+\varkappa)\ell_j^T\mathsf D\ell_j\dd\tau$, $\dd\langle a_j,a_k\rangle=4\ell_j^T\mathsf D\ell_k\dd\tau$. It\^o's formula for $p_j=e^{a_j}/\sum_ke^{a_k}$, with $\partial p_j/\partial a_k=p_j(\delta_{jk}-p_k)$ and $\partial^2p_j/\partial a_k\partial a_l=p_j(\delta_{jk}-p_k)(\delta_{jl}-p_l)-p_jp_k(\delta_{kl}-p_l)$, gives
\[
\frac{\E[\dd p_j\mid\mathcal F_\tau]}{p_j\dd\tau}=2(\ell_j-\bar\ell)\cdot\mathsf b_\tau-(1+\varkappa)\bigl(\ell_j^T\mathsf D\ell_j-\langle\ell^T\mathsf D\ell\rangle\bigr)+2\ell_j^T\mathsf D\ell_j-4\ell_j^T\mathsf D\bar\ell+4\bar\ell^T\mathsf D\bar\ell-2\langle\ell^T\mathsf D\ell\rangle,
\]
which rearranges to \eqref{eq:history_drift}. (ii) Write $\mathsf b_\tau=2\mathsf D\bar\ell+c_\tau$. The bracket in \eqref{eq:history_drift} is $\varpi(\ell_j)-\sum_kp_k\varpi(\ell_k)$ with $\varpi(\ell)=2\ell\cdot c_\tau+(1-\varkappa)\ell^T\mathsf D\ell$, which vanishes for all $j$ if and only if $\varpi$ takes one value on all records. The linear part $2\ell\cdot c_\tau$ ranges over all linear functions of $\ell$, so $\varpi$ can be made constant on the records exactly when $(1-\varkappa)\ell^T\mathsf D\ell$ coincides on the records with some affine function, which is the criterion stated. If the records are affinely independent, every function on them is the restriction of an affine function, and a constant $c$ absorbs any $\varkappa$. If three records are collinear, $\ell_0+t_jd$, the restriction of $\varpi$ is a quadratic polynomial in $t$ with leading coefficient $(1-\varkappa)d^T\mathsf Dd$ taking one value at three points, so $\varkappa=1$ and then $c_\tau\cdot d=0$. The Stratonovich law $\dd\phi=\sum_aL_a\phi\circ\dd Y_a$, the Wong--Zakai limit of a smooth drive, has It\^o form \eqref{eq:linear_history} with $\varkappa=-1$ whatever the drift of $Y$, since the It\^o correction $\frac12\sum_{ab}\mathsf D_{ab}L_aL_b$ depends only on the quadratic covariation. (iii) Under $\ell_j\mapsto\ell_j+s$ the bracket in \eqref{eq:history_drift} changes by $2(\ell_j-\bar\ell)^T\mathsf Ds\,(2\alpha-1-\varkappa)$, which vanishes identically only for $\varkappa=2\alpha-1$; substituting $\varkappa=2\alpha-1$ and $\mathsf b=2\alpha\mathsf D\bar\ell$ and completing the square gives \eqref{eq:fractional_drift}. For two records $\ell_1-\bar\ell=\delta(1-p)$, $\ell_0-\bar\ell=-\delta p$ and $\langle(\ell-\bar\ell)^T\mathsf D(\ell-\bar\ell)\rangle=\delta^T\mathsf D\delta\,p(1-p)$. The log-odds $y=\log[p/(1-p)]$ then obey $\dd y=2\delta\cdot\dd\widetilde W+2\alpha\,\delta^T\mathsf D\delta\tanh(y/2)\dd\tau$ for any zero of $L$, with $\widetilde W$ of covariance $\mathsf D\dd\tau$ so that $\operatorname{Var}\dd y=4\delta^T\mathsf D\delta\,\dd\tau$; the scale density is $\operatorname{sech}^{2\alpha}(y/2)$, which is integrable at both ends for $\alpha>0$, giving transience with exit law $I_{p_0}(\alpha,\alpha)$ as in Theorem~2.4 of Ref.~\cite{projectvii} (their $\Gamma$ is $4\delta^T\mathsf D\delta$), and non-integrable for $\alpha\leq0$, giving recurrence. (iv) Record eigenstates are fixed points of \eqref{eq:linear_history} up to normalisation, so for them $\E p_j$ is constant; for a superposition with $0<p<1$ the drift in (iii) is nonzero unless $p=\frac12$; the rest is the argument around \eqref{eq:ensemble}.
\end{proof}

Three remarks. First, the proposition says what Proposition~\ref{prop:bare} said, now from the side a derivation sees: a norm-preserving law is one whose history has \emph{some} conditional law, and the Born rule is a statement about that law---the history must lean towards the leading record with gain $2\mathsf D$, and, when the records can see it, the linear law must carry its counter-term with unit coefficient. Second, the two ingredients are separately necessary and are tested by different experiments. The counter-term is invisible whenever $\ell_j^T\mathsf D\ell_j$ is the same for all records, as it is for two records with symmetric eigenvalues; a derivation checked only on such a case can have $\varkappa=-1$ and pass, and a history produced by a smooth coefficient has $\varkappa=-1$ whatever its gain. The ancillary script shows this on the Einstein locus of Section~\ref{sec:einstein}: with the gain right and $\varkappa=-1$, a projector $L$ with eigenvalues $1,0$ selects its eigenvalue-one record with probability one from any initial weights (the log-odds drift $4\mathsf Dp$ vanishes as the other record takes over, which makes that side recurrent), and for three records at $\ell=1,0,-1$ the combination $a_1+a_{-1}-2a_0$ of log-weights is conserved, so from weights $(0.3,0.4,0.3)$ the middle weight never exceeds $0.4$ and that record is never selected; restoring $\varkappa=1$ restores Born to within $0.03$ in the simulation. Third, the fractional law is not a harmless intermediate. By (iv) it is an ensemble nonlinearity, and Section~\ref{sec:failures} turns such nonlinearities into signals; so within this class no-signalling, linearity on ensembles and the Born gain $\alpha=1$ are one condition, which is the triangle stated in Ref.~\cite{gtdborn} with the gain now identified. The script verifies \eqref{eq:history_drift} by exact one-step quadrature for three records and the exit law $I_{p_0}(\alpha,\alpha)$ for $\alpha=\frac12,1,2$ to within $0.005$. Section~\ref{sec:einstein} asks what fixes $\alpha$ when the history comes from a sea.

\subsection{Which length? The geodesic proposal and the indefinite form}\label{sec:geodesicnorm}
There is a proposed geometric reason for norm conservation in GTD: the emergent state vector is compared with the velocity of a geodesic in matrix configuration space \cite[Sec.~4]{gtdborn}. The geometric fact behind this is exact. If $U=\dd Q/\dd\tau$ is an affinely parametrised geodesic velocity for a metric-compatible connection, then
\begin{equation}
\nabla_UU=0,\quad\nabla g=0\quad\Longrightarrow\quad
\frac{\dd}{\dd\tau}g(U,U)=2g(\nabla_UU,U)=0.
\label{eq:geodesic_norm}
\end{equation}
If a physical projection identified a positive version of $g(U,U)$ with $\|\psi\|^2$, normalisation would have a geometric origin. Section~\ref{sec:normdrift} has already shown that this would not settle how the normalised state moves: the Born-restoring response is tangent and invisible to the conserved length. A second difficulty concerns the length itself.

Suppose the configuration-space metric of the geodesic proposal is built from the split form \eqref{eq:splitnorm}, the one that carries the spacetime interpretation of Section~\ref{sec:splitform}; if it is built from the positive form instead, the connection with that interpretation is lost. Then $g(U,U)$ is indefinite: negative on the open cone $|u|>|x|$, zero on the null cone, positive on the open cone $|x|>|u|$. Geodesic flow preserves the value and hence the sign, so the positive cone is flow-invariant and carries a conserved positive length without further choices. That is not enough for a probability norm. A quantum norm is a positive-definite quadratic form on a \emph{linear} space of states: superpositions must have positive norm, and the weight of each component must be nonnegative. The positive cone is not closed under superposition, and a vector in it can have components of negative length---$g(U_1,U_1)=-1$, $g(U_2,U_2)=4$, $g(U_1,U_2)=0$ gives $g(U_1+U_2,U_1+U_2)=3$ with a component of weight $-1$. A positive-definite restriction of $g$ exists only on a subspace of at most half the dimension (three of the six directions of $V$, and correspondingly for matrix-valued velocities), such a subspace is not invariant under $SO(3,3)$ (its stabiliser is $S(O(3)\times O(3))$), and it is not preserved by the geodesic flow of a general metric unless the restriction is imposed as a constraint. Selecting it is a selection of the same kind as picking the Lorentzian leaf $-u_1^2+|x|^2$ in Section~\ref{sec:splitform}: the spatial directions of a leaf span a maximal positive subspace. It is natural to conjecture that the two selections are one and the same emergence step; if so, norm conservation cannot be used as an input to the emergence argument, because it is one of its outputs. The alternative is to build the metric from a positive trace form such as $\Tr(U^\dagger U)$; that length is positive and conserved for free motion, but it is invariant under the unitary group rather than $SO(3,3)$, it has no relation to the split geometry, and its conservation in the interacting theory is what is in question \cite[Secs.~7.3.6, 9.2.1]{book}. Either way the geodesic length is a candidate explanation to be completed by the projection, not a replacement for it. A conserved quantity for the full system also assigns no probabilities to conditioned subsystem histories; the correlator calculation of Ref.~\cite{correlator} supplies a coloured two-point function, not the joint law.

\section{Causal consistency of emergent records}\label{sec:causal}
We now give the abstract labels $j$ an experimental meaning. This assumes that approximately classical detectors and a laboratory spacetime are available as an effective description, as discussed in Section~\ref{sec:connes}; the microscopic origin of the detectors is not at issue here.

\subsection{One record from many detectors}\label{sec:solvay}
Model each detector's record register by two states $|0\rangle_i$, $|1\rangle_i$ (no absorption, absorption) and define on the full tensor product of registers the local projector
\begin{equation}
\Pi_i=\id_1\otimes\cdots\otimes|1\rangle_i\langle1|\otimes\cdots\otimes\id_n.
\label{eq:localN}
\end{equation}
After ideal premeasurement the state lies in the one-record sector,
\begin{equation}
|\Psi\rangle=\sum_{j=1}^nc_j|D_j\rangle,\qquad |D_j\rangle=|0\cdots1_j\cdots0\rangle,\qquad \Pi_i|D_j\rangle=\delta_{ij}|D_j\rangle,
\label{eq:solvaystate}
\end{equation}
with other degrees of freedom absorbed into the orthogonal $D_j$. The source and the premeasurement interaction supply this sector; the collapse law does not derive it, and detector efficiencies, dark-count branches and angular amplitudes belong to the preparation. Defining $\Pi_i$ on the full apparatus space matters: the branch projector $|D_i\rangle\langle D_i|$ is a global operator, and its agreement with $\Pi_i$ inside the one-record sector does not make it local on arbitrary apparatus states.

The division into premeasurement followed by reduction is a calculational approximation, and it is where a GTD calculation would differ from CSL: in a universal collapse dynamics the same law acts while the signal is being amplified, when the measurement interaction does not commute with the eventual pointer projectors, so the incident field, the detector constituents and the collapse channels must be evolved together, as mass-proportional CSL does through the detector matter \cite{csl,bassi}. The one-record sector describes the later stage in which the alternatives have become stable.

Let $s_i>0$ be the accumulated exposure of channel $i$---the dimensionless integrated collapse strength---and, with respect to Lebesgue measure $\dd y_i$,
\begin{equation}
K_i(y_i,s_i)=(2\pi s_i)^{-1/4}\exp\Bigl[-\frac{(y_i\id-2s_i\Pi_i)^2}{4s_i}\Bigr],\qquad \int K_i^\dagger K_i\,\dd y_i=\id .
\label{eq:localK}
\end{equation}
If register $i$ is in $|0\rangle_i$ the history $y_i$ is Gaussian with mean $0$ and variance $s_i$; if it is in $|1\rangle_i$ the mean is $2s_i$. For independent channels the physical joint density and posterior weights are
\begin{align}
p(\boldsymbol y\mid\boldsymbol s)&=\sum_j|c_j|^2\prod_i\frac{e^{-(y_i-2s_i\delta_{ij})^2/(2s_i)}}{\sqrt{2\pi s_i}},\label{eq:solvaydensity}\\
p_j(\boldsymbol y)&=\frac{|c_j|^2e^{2y_j-2s_j}}{\sum_k|c_k|^2e^{2y_k-2s_k}}.\label{eq:solvayposterior}
\end{align}
As the exposures grow, a history consistent with one detector's absorption makes that alternative dominate, and the other weights decrease through the same normalisation; there is no term describing a message from the dominant detector to the rest. For each occupied pair $j,k$ the distinguishing exposure is $s_j+s_k$, and if all pairwise exposures grow without bound, \eqref{eq:selection} gives asymptotic selection. At finite exposure,
\begin{equation}
\E_{\mathbb P}\bigl(1-\max_jp_j\bigr)\leq\sum_{j<k}\sqrt{p_j^0p_k^0}\,e^{-(s_j+s_k)/2},\qquad p_j^0=|c_j|^2,
\label{eq:error}
\end{equation}
proved in Appendix~\ref{app:bound}. For equal weights and exposures the bound is $(n-1)e^{-s}/2$, so $s\geq\log[(n-1)/(2\varepsilon)]$ makes it at most $\varepsilon$. The bound controls the mean residual ambiguity; it does not make the small surviving amplitudes vanish. That is the tails problem common to all Gaussian collapse models \cite{albertloewer,lewistails,tails}, and the model resolves Einstein's problem in the sense in which CSL does---definite enough, by a tolerance appropriate to the apparatus---not by exact prompt projection.

A channel coupling identically to every alternative multiplies them all by the same random factor and selects nothing, while a shared channel with different eigenvalues can select, so independent noise per detector is a sufficient construction, not a necessary one. Distinguishing $f_j$ in $\sum_jc_j|f_j\rangle|g_j\rangle$ also selects the correlated geometric alternative $g_j$ without removing the uncertainty within it; the record model tests a proposed effective law once an approximately stable background is available, and does not claim that each click creates a spacetime.

\subsection{The no-signalling theorem}\label{sec:nosig}
Selecting a record and obeying a light-cone constraint are different tests. Assume that a time-oriented Lorentzian spacetime and two spacelike laboratory regions have been identified. Alice and Bob choose controls $a,b$ and record outcomes $\alpha,\beta$. No-signalling from Alice to Bob means
\begin{equation}
p(\beta\mid a,b)=p(\beta\mid b)\qquad\text{for every allowed }a,\label{eq:NSdefinition}
\end{equation}
where Bob uses all his runs; a probability conditioned on Alice's outcome need not satisfy it, since learning which runs have that outcome requires the outcome to be communicated. By causal consistency we mean this condition together with independence of the ordering of spacelike operations. We do not mean Bell's local-causality factorisation, a derived propagation mechanism in pregeometry, or a relativistic field theory of collapse.

The theorem that follows is the standard one \cite{grw80,beckman}; we state it because the rest of the paper is about its hypotheses. Let the laboratory Hilbert factors include the local pointers and any local memories, and let $A_{\alpha|a}$, $B_{\beta|b}$ be the operators of the local operations, with any additional Kraus index absorbed into the integration. Each apparatus must account for all its outcomes,
\begin{equation}
\int A_{\alpha|a}^\dagger A_{\alpha|a}\dd\alpha=\id_A,\qquad
\int B_{\beta|b}^\dagger B_{\beta|b}\dd\beta=\id_B .\label{eq:localcomplete}
\end{equation}
For a joint preparation $\rho_{AB}$ independent of the later choices, the instrument rule gives $p(\alpha,\beta\mid a,b)=\Tr[\rho_{AB}(A_{\alpha|a}^\dagger A_{\alpha|a}\otimes B_{\beta|b}^\dagger B_{\beta|b})]$, and integrating over Alice's outcome replaces her factor by the identity:
\begin{equation}
p(\beta\mid a,b)=\Tr[B_{\beta|b}\rho_BB_{\beta|b}^\dagger]=p(\beta\mid b),\qquad\rho_B=\Tr_A\rho_{AB}.
\label{eq:nosignal}
\end{equation}
This holds for entangled as well as product preparations, and uses the Born instrument rule, completeness and locality; a null cone alone would not imply it. For a finite sequence of local operations one integrates the last local outcome first and repeats the argument backwards; a laboratory's controls may depend on its own earlier records provided each conditional operation is complete. The identity therefore holds for the whole remote record distribution, and it survives summing over inaccessible local variables, so a detector that exposes only a coarse function of $y_i$, or no noise readout at all, is covered. The Solvay channels \eqref{eq:localK} are complete local operations on their registers, so Solvay records are non-signalling. Nothing in this is specific to GTD: the conclusion is inherited by any law of the form \eqref{eq:SSE} with fixed local operators. The identification of an actual GTD fluctuation history with such instruments is the assumption to be checked, and Section~\ref{sec:failures} shows how it can fail and where it cannot.

\subsection{Where state dependence may enter}\label{sec:failures}
The ensemble average of \eqref{eq:SSE} for fixed bounded $L_a$ is the linear Lindblad equation
\begin{equation}
\frac{\dd\rho}{\dd\tau}=-\frac{i}{\hbar}[H,\rho]-\frac12\sum_a[L_a,[L_a,\rho]],\qquad\rho=\E_{\mathbb P}|\psi\rangle\langle\psi|,
\label{eq:ensemble}
\end{equation}
so two ensembles with the same density matrix evolve identically on average. This linearity is the protection against the signalling mechanism of Gisin and Polchinski \cite{gisin89,polchinski,simongisin}. The deterministic law \eqref{eq:detnorm} does not have it: with $H_S=0$ and $A=a_0\sigma_z$ the equal mixture of $\sigma_z$ eigenstates stays $\id/2$ while the equal mixture of $\sigma_x$ eigenstates, initially the same $\id/2$, becomes $\frac12[\id+\tanh(2a_0\tau)\sigma_z]$, so remote steering between the two preparations would be detectable. Linearity is still not the whole requirement: the completely positive generator with $L=\sqrt\gamma\,|1\rangle_A\langle1|\otimes X_B$ leaves Bob alone when Alice prepares $|0\rangle$ and subjects him to bit-flip noise when she prepares $|1\rangle$; stipulated across spacelike regions it signals, and calling its origin pregeometric does not change that.

GTD proposes that collapse is enhanced in a regime of large entanglement \cite{gtdcritical,book}. The most direct way to write this---a rate that is a function of the entanglement of the conditioned global state, multiplying the collapse operator---puts the state dependence into the ensemble generator, and the failure is quantitative.

\begin{proposition}[A rate evaluated on the conditioned global state signals]\label{prop:ent}
Let an effective two-qubit pure state obey the CSL law \eqref{eq:SSE} with $H=0$ and the single operator
\begin{equation}
L_A(\Psi)=\sqrt{\gamma(E(\Psi))}\,Z_A\otimes\id_B,
\label{eq:ent_operator}
\end{equation}
where $E(\Psi)$ is the entanglement entropy of the conditioned joint state, $\gamma\geq0$, and the drift is the CSL drift evaluated at that coefficient. Let Bob's measurement be ideal and projective with Born probabilities. Let $0<p<1$, $p\neq\frac12$, and $h(p)=-p\log p-(1-p)\log(1-p)$. If $\gamma(h(p))\neq\gamma(0)$, Bob's unread measurement changes Alice's subsequent local statistics.
\end{proposition}
\begin{proof}
Prepare $|\Psi_p\rangle=\sqrt p\,|+\rangle_A|0\rangle_B+\sqrt{1-p}\,|-\rangle_A|1\rangle_B$, with $|\pm\rangle$ the $X_A$ eigenstates, so that $\rho_A=\frac12[\id+(2p-1)X_A]$ and $E=h(p)$. Averaging the It\^o product rule for $|\Psi\rangle\langle\Psi|$ over the next increment gives the initial drift $\dot\rho_A=\gamma(h(p))(Z_A\rho_AZ_A-\rho_A)$, hence $\dd\langle X_A\rangle/\dd\tau=-2(2p-1)\gamma(h(p))$. If Bob measures in the $|0\rangle,|1\rangle$ basis, each conditioned branch is a product state with $E=0$; averaging over his unread outcome returns the same $\rho_A$ but with rate $\gamma(0)$ on both branches, so $\dd\langle X_A\rangle/\dd\tau=-2(2p-1)\gamma(0)$. Different slopes are distinguished by repeated measurements at Alice. For $p=0.8$ the two coefficients are $-1.2\,\gamma(h(0.8))$ and $-1.2\,\gamma(0)$.
\end{proof}

\begin{remark}[Other functionals, and Bob's channel]\label{rem:ent}
Nothing depends on entropy in particular: Bob's measurement replaces one conditioned state by a convex decomposition of it with the same $\rho_A$, and the argument goes through for any rate functional whose value on the mixed state differs from its branch-averaged value weighted by the action of the generator, as it does for purity. Modelling Bob's measurement by his own local collapse channel at large exposure changes the slopes by exponentially small amounts and not the conclusion. What is tested is a physical remote operation with branch selection, not a distant observer's decision to learn or forget something.
\end{remark}

The proposition is a counterexample to one way of writing an entanglement-dependent rate. It is not a theorem that no such rate exists, and the reason is the freedom noted at the end of Section~\ref{sec:normdrift}. Bassi, D\"urr and Hinrichs \cite{uniqueness} showed that a diffusive, norm-preserving, Markovian state-vector law whose ensemble average is a completely positive linear map must have noise fields of the form
\begin{equation}
B_a(\psi)=\bigl(L_a^\psi-\ell_a^\psi\bigr)\psi,\qquad
L_a^\psi=\sum_bu_{ab}(\psi)L_b,\qquad
\ell_a^\psi=\tfrac12\langle\psi,(L_a^\psi+L_a^{\psi\dagger})\psi\rangle,
\label{eq:bdh_form}
\end{equation}
with fixed Lindblad operators $L_b$ and a unitary matrix $u(\psi)$ that may depend on the state, together with the matching drift, which for self-adjoint $L_b$ is $f=[-\frac12\sum_bL_b^2+\sum_a\ell_a^\psi L_a^\psi-\frac12\sum_a(\ell_a^\psi)^2]\psi$. Complete positivity does not follow from their other assumptions, as Di\'osi showed \cite{diosicomment,bdhreply}; we add it as a hypothesis, which no-signalling with entangled ancillas supplies \cite{simongisin}. Jump unravellings and coloured noise are outside their hypotheses. Within them the only state dependence that survives is the unravelling, and they note themselves that the admissible diffusions range from the random-phase law to the CSL law, with the collapse rate maximal for the latter; the Wiseman--Di\'osi classification \cite{wiseman} parametrises the same freedom for fixed $u$. The following proposition states this for record weights with several channels, with the equality condition, and with the proviso under which it is a no-signalling statement.

\begin{proposition}[What a state-dependent unravelling can and cannot do]\label{prop:bound}
Let the $L_b$ be commuting self-adjoint operators with eigenvalues $\ell_{bj}$ on the records $|j\rangle$ and eigenvalue vectors $\ell_j$, and let the state-vector law have the form \eqref{eq:bdh_form} with any unitary $u(\psi)$. Then every ensemble prediction is that of the Lindblad equation \eqref{eq:ensemble}, and the conditional quadratic variation of each record weight $p_j=|\langle j|\psi\rangle|^2$ satisfies
\begin{equation}
\frac{\E[(\dd p_j)^2\mid\psi]}{\dd\tau}=4p_j^2\sum_a\Bigl(\operatorname{Re}\sum_bu_{ab}(\psi)\bigl(\ell_{bj}-\bar\ell_b\bigr)\Bigr)^2\;\leq\;4p_j^2\,|\ell_j-\bar\ell|^2,
\label{eq:selection_bound}
\end{equation}
with equality for the CSL law $u=\id$, and more generally if and only if $\sum_bu_{ab}(\psi)(\ell_{bj}-\bar\ell_b)$ is real for every $a$; the quadratic variation vanishes for $u=i\,\id$, the random-phase law. Hence a state-dependent unravelling can make the realised selection rate depend on the conditioned state, including on its entanglement, without changing any ensemble prediction; but it can only reduce that rate below the value fixed by the operators. The law is then no longer an instrument in the sense of Section~\ref{sec:joint}, because its increment law is not linear in $\rho$, and the no-signalling statement holds at the level of the density matrix only: if the fluctuation history were itself a local record, a remote choice that changes $u(\psi)$ would be readable in the correlation between that record and the local weights.
\end{proposition}
\begin{proof}
The ensemble generator of \eqref{eq:bdh_form} with its matching drift is $\sum_b(L_b\rho L_b-\frac12\{L_b^2,\rho\})$ for every $u$, by direct substitution, since $\sum_aL_a^\psi\rho L_a^{\psi\dagger}=\sum_bL_b\rho L_b$ by unitarity. In the record basis the noise term changes the amplitude $c_j$ by $(\sum_bu_{ab}\ell_{bj}-\ell_a^\psi)c_j\dd W_a$, and since $\ell_a^\psi=\operatorname{Re}\sum_bu_{ab}\bar\ell_b$ the weight changes by $\dd p_j=2p_j\sum_a\operatorname{Re}[z_{aj}]\dd W_a$ with $z_{aj}=\sum_bu_{ab}(\ell_{bj}-\bar\ell_b)$. The quadratic variation is $4p_j^2\sum_a(\operatorname{Re}z_{aj})^2\leq4p_j^2\sum_a|z_{aj}|^2=4p_j^2|\ell_j-\bar\ell|^2$ by unitarity of $u$, with equality if and only if every $z_{aj}$ is real. The last statement follows from $\E[\dd W_a\,\dd p_j\mid\psi]=2p_j\operatorname{Re}z_{aj}\dd\tau$, which depends on $u(\psi)$.
\end{proof}

The ancillary script illustrates the proposition on the state of Proposition~\ref{prop:ent} with a single channel $L=Z_A$ and the phase $u=e^{i\phi(\psi)}$, $\phi=\frac\pi2(1-E(\psi)/\log2)$, so that selection is switched on by entanglement: Bob's unread measurement leaves the slope of $\langle X_A\rangle$ unchanged at $\mp1.2\gamma$ while switching the quadratic variation of Alice's weights from $0.82\gamma$ to zero. Two features of this example deserve notice. Whether a local record becomes definite then depends on a remote choice, even though no local statistic does; and an unravelling switched on by entanglement quenches itself, since selection removes the entanglement that switches it on, which lengthens the tails.

\begin{example}[Operator-carried amplification]\label{ex:amplification}
Let $N$ two-level aikyons carry record projectors $n_i=|1\rangle_i\langle1|$ and let the single collapse operator be the collective $L=\sqrt{\gamma_0}\sum_{i=1}^Nn_i$. For the branch pair $|0\rangle^{\otimes N}$ versus $|1\rangle^{\otimes N}$ the eigenvalue difference is $\sqrt{\gamma_0}N$, so by \eqref{eq:selection} the log-ratio of the two branches drifts at rate $2\gamma_0N^2$; two product-basis records differing in one aikyon have eigenvalue difference $\sqrt{\gamma_0}$ and are resolved at rate $2\gamma_0$. The $N^2$ enhancement is a property of the macroscopic distinctness of the two branches as seen by $L$---the same pair separates at the same rate inside a product state---and it is carried entirely by the operator. It is the amplification mechanism of CSL \cite{csl,bassi} restated for GTD alternatives, and it lies inside the no-signalling class with the maximal selection rate of \eqref{eq:selection_bound}. In this admissible and observable form, what GTD calls entanglement enhancement is enhancement by macroscopic distinctness. The same bookkeeping is done for correlated constituents in Refs.~\cite{microclosure,physrecon}: for $L=\sum_iL_i$ with equicorrelated couplings the rate scales as $N$, $N^2$ or zero according to the correlation, so an $N^2$ exponent by itself witnesses neither entanglement nor collapse; the exponent is a property of the operator and the state together.
\end{example}

The design rule that follows is sharper than ``entanglement cannot enter''. Within the Markovian diffusive class, state dependence of a collapse rate may enter in two places. In the unravelling, as in Proposition~\ref{prop:bound}, it does not signal provided the fluctuation history is not a local record, but it can only switch selection off and on below the operator-fixed rate, and it has no ensemble consequences: no interference-decay or heating experiment could see it, and it would not alter a single CSL bound. In the operators, as in Example~\ref{ex:amplification}, it raises the rate and is testable, but then it is a statement about the spectrum of the collective observable that the unresolved aikyons couple to, which must spread with the number of constituents in the superposition, and it is a dependence on macroscopic distinctness rather than on entanglement as such. A scalar rate multiplying the generator, as in Proposition~\ref{prop:ent}, is excluded by the uniqueness theorem, and so, by Proposition~\ref{prop:history}(iv), is a response gain other than the Born one: a sea that leans towards the leading record with the wrong gain is a state-dependent rate in disguise. If GTD's enhancement is to be both causally admissible and observable within this class, it must be derived in the second form; a non-Markovian emergent law, which Section~\ref{sec:einstein} shows to be the generic case for a sea with finite correlation time, would have to be analysed on its own terms.

\subsection{Spacelike composition, Bell pairs, and the continuum}\label{sec:continuum}
Operations on different laboratory factors commute, so exchanging their order leaves the joint unnormalised state and its probability unchanged, provided the entire supports of the operations are spacelike. The detector model admits covariant apparatus parameters, for instance $s_i=\int\kappa_i\,\dd\tau_i^{\rm proper}$ with prescribed local scalar rates, so that a change of coordinates describes the same exposures; this is a compatible assignment within the model, not a derivation of proper times from $\tau$. There is also a temporal limitation: the selection theorem takes exposure to infinity, and monitoring intervals of infinite duration are not mutually spacelike. The selection theorem establishes the limiting mechanism, \eqref{eq:nosignal} the causal test for finite operations, and neither treats an infinite experiment as one spacelike event.

For a spin-singlet pair \cite{epr,bell} with local measurement directions $\mathbf a,\mathbf b$ the pointer weights are $p^0_{rt}=(1-rt\,\mathbf a\cdot\mathbf b)/4$, $r,t=\pm1$. Applying the local channels \eqref{eq:localK} to the two pointer registers selects the four alternatives under the same law, with the quantum correlations as inputs, and the no-signalling theorem keeps both marginals at $1/2$. If each wing reports the sign of $y_i-s_i$---a readout of the fluctuation history that is covered by the theorem for these fixed-operator channels, though not for the state-dependent unravellings of Proposition~\ref{prop:bound}---the readout reliability is $w(s)=\operatorname{erf}\sqrt{s/2}$, the reported correlation is $-w(s_A)w(s_B)\,\mathbf a\cdot\mathbf b$, and the CHSH combination is $-2\sqrt2\,w(s_A)w(s_B)$, approaching $-2\sqrt2$ as both exposures grow; every run is counted. This preserves the given Bell correlations and is not a Bell-local resolution of the EPR argument: the latent $j^\star$ of \eqref{eq:selection} represents each fixed-setting history law without assigning predetermined outcomes across settings.

Two structural questions remain at the continuum level. Local quantum-field algebras need not factor into finite detector Hilbert spaces, and continuum reduction brings its own requirements on domains, energy production and integrability \cite{fewsterverch,ggp,tumulka2020}. More specifically, Myrvold has shown that a relativistically invariant Markovian collapse theory on standard fields, with a stable vacuum and the forward-light-cone spectrum condition, cannot be built on the standard state space \cite{myrvold}. A GTD completion therefore faces a choice. Either the emergent description is Lorentz invariant and Markovian, in which case the theorem applies and the nonstandard degrees of freedom it demands must be exhibited---the aikyon geometric variables and the auxiliary sector of Appendix~\ref{app:micro} are the candidates, and it must be shown that they are not standard fields obeying the spectrum condition; or the preferred $\tau$-ordering of Section~\ref{sec:connes} survives as a preferred foliation, which evades the theorem at the price of a Lorentz violation whose size the completion must predict; or the emergent law is not Markovian and lies outside the theorem's hypotheses. The third option is not merely available. Proposition~\ref{prop:einstein} shows that a sea with a finite correlation time $\tau_c$ produces a history drift with memory, so the Markov law of Section~\ref{sec:joint} is the limit $\tau_c\ll$ selection time of a law that has memory at shorter scales; whether that memory survives at the scale of laboratory operations is a quantitative question about $\tau_c$, not a choice to be made, and memory confined to scales shorter than those operations does not by itself remove the law from the theorem's hypotheses at the scale where they are applied. Our finite registers have no field-theoretic vacuum, so their consistency test does not decide among these.

\section{From GTD coefficients to collapse noise}\label{sec:diagnostics}
The effective law exists. The question is whether GTD produces it, and the way to answer is to follow the quantity that is supposed to supply the fluctuations as far as the construction allows. Two checks must be kept apart: does a candidate coefficient fluctuate in $\tau$, rather than merely differ between preparations, and does the proposed map to physical states turn it into a nonzero anti-self-adjoint operator? This section answers the first question for a model whose relation to the Grassmann block is stated, and shows where the second stands.

\subsection{The free coefficient is frozen}\label{sec:microcalc}
Appendix~\ref{app:micro} records the action. For the free fundamental fermionic coordinate it implies
\begin{equation}
D_+^2q_F=0,\qquad D_+=\partial_\tau+i\nu,\qquad q_F(\tau)=e^{-i\nu\tau}(q_{F0}+\tau V_F),
\label{eq:free_solution}
\end{equation}
so $D_+q_F=e^{-i\nu\tau}V_F$ and the bifermionic Lagrangian coefficient---the term of the action bilinear in the fermionic velocities---with the odd auxiliary factors in fixed order, is
\begin{equation}
\mathcal L_{FF}=-c_0\eta^3\,b\,\Tr(V_F^\dagger V_F),\qquad b=\beta_1\beta_2 .
\label{eq:free_aux}
\end{equation}
The phases cancel and $\Tr(V_F^\dagger V_F)$ is a constant of the free motion. To obtain a history law rather than inspect an equation, restrict to one scalar block and two dynamical Grassmann generators $\theta_1,\theta_2$, distinct from the auxiliary $\beta$'s, with $\theta_r^\dagger=\theta_r$, $q_{F0}=0$ and
\begin{equation}
V_F=\theta_1v_1+\theta_2v_2,\qquad V_F^\dagger V_F=iJ\theta_1\theta_2,\qquad J=2\operatorname{Im}(v_1^*v_2),
\label{eq:free_J}
\end{equation}
the randomness residing in the ordinary complex coefficients $v_r$, not in the symbols. Preparing $v_1,v_2$ as independent circular complex Gaussians of variance $\sigma^2$ gives a Laplace law for the coefficient,
\begin{equation}
\E e^{ikJ}=\frac1{1+\sigma^4k^2},\qquad f_J(j)=\frac{e^{-|j|/\sigma^2}}{2\sigma^2},\qquad\operatorname{Var}J=2\sigma^4 ,
\label{eq:laplace}
\end{equation}
and since $J(\tau)=J(0)$ on every trajectory the integrated coefficient $Y(\tau)=\int_0^\tau J\,\dd s=J\tau$ has $\operatorname{Cov}(Y(t),Y(u))=2\sigma^4tu$: a static random offset whose variance grows as $\tau^2$. In a free theory every function of the velocities is conserved, so this was certain; its value is to set the baseline against which an interaction is judged.

\subsection{Coefficient noise from a continuum of unresolved modes}\label{sec:bath}
What matters for the coefficient is the joint statistics of the velocities $v_1,v_2$ as functions of $\tau$. Write the four real components as a vector $\mathbf v=(\operatorname{Re}v_1,\operatorname{Im}v_1,\operatorname{Re}v_2,\operatorname{Im}v_2)$; then
\begin{equation}
J=\mathbf v^T\mathsf M\mathbf v,\qquad \mathsf M=\begin{pmatrix}0&\epsilon\\-\epsilon&0\end{pmatrix},\quad\epsilon=\begin{pmatrix}0&1\\-1&0\end{pmatrix},
\end{equation}
a quadratic form of signature $(2,2)$. Suppose that, under some interacting dynamics with some preparation, $\mathbf v(\tau)$ is a stationary Gaussian process with matrix correlation $\mathsf C(t)=\langle\mathbf v(t)\mathbf v(0)^T\rangle$ and spectral density matrix $\mathsf S(\omega)=\int\mathsf C(t)e^{i\omega t}\dd t$, which is Hermitian and positive semidefinite. Then Wick's theorem gives the statistics of $J$ in closed form.

\begin{proposition}[Coefficient noise from stationary Gaussian velocities]\label{prop:bath}
If $\mathbf v$ is a stationary Gaussian process with mean zero, then
\begin{align}
\langle J\rangle&=\Tr[\mathsf M\mathsf C(0)],\qquad
\operatorname{Cov}\bigl(J(t),J(0)\bigr)=2\Tr\bigl[\mathsf M\mathsf C(t)\mathsf M^T\mathsf C(t)^T\bigr],\notag\\
D_J&:=\lim_{\tau\to\infty}\frac{\operatorname{Var}Y(\tau)}{\tau}=\frac1\pi\int_{-\infty}^\infty\Tr\bigl[\mathsf M\mathsf S(\omega)\mathsf M\mathsf S(\omega)\bigr]\dd\omega ,
\label{eq:JJ_general}
\end{align}
where the integrand equals $\|\mathsf S^{1/2}\mathsf M\mathsf S^{1/2}\|_F^2\geq0$, so that $D_J=\infty$ whenever the integral diverges. In particular, if $v_1$ and $v_2$ are independent circular complex processes with $C_r(t)=\langle v_r(t)v_r^*(0)\rangle$ and scalar spectra $S_r\geq0$,
\begin{equation}
\langle J\rangle=0,\qquad
\langle J(t)J(0)\rangle=2\operatorname{Re}\bigl[C_1(t)C_2(t)^*\bigr],\qquad
D_J=\frac1\pi\int_{-\infty}^\infty S_1S_2\,\dd\omega ,
\label{eq:JJ}
\end{equation}
and, for absolutely continuous spectra, $0<D_J<\infty$ if and only if $S_1S_2$ is integrable and not almost everywhere zero. The one-time law of $J$ is then the Laplace law \eqref{eq:laplace} with $\sigma^2\to\sigma_1\sigma_2$, whatever the dynamics. For the free theory $S_r=2\pi\sigma_r^2\delta(\omega)$ and $D_J=\infty$ (ballistic growth); for any finite bath the spectra are sums of lines, and $D_J$ is infinite if the two spectra share a line and zero otherwise, never finite and positive; a continuum of modes with integrable $S_1S_2$ gives diffusion.
\end{proposition}
\begin{proof}
For a centred Gaussian vector process, Wick's theorem gives $\operatorname{Cov}(\mathbf v^T\mathsf M\mathbf v(t),\mathbf v^T\mathsf M\mathbf v(0))=\Tr[\mathsf M\mathsf C(t)\mathsf M^T\mathsf C(t)^T]+\Tr[\mathsf M\mathsf C(t)\mathsf M\mathsf C(t)^T]$, and $\mathsf M$ is symmetric, so the two terms are equal. Parseval's identity, with $\mathsf C(-t)=\mathsf C(t)^T$ so that $\mathsf S(-\omega)^T=\mathsf S(\omega)^\dagger=\mathsf S(\omega)$, converts the time integral of the covariance into the frequency integral, and $\operatorname{Var}Y(\tau)=2\int_0^\tau(\tau-s)\operatorname{Cov}(J(s),J(0))\dd s$ gives the limit. For independent circular components the pairings $\langle v_1^*v_1^*\rangle\langle v_2v_2\rangle$ vanish and the general formula reduces to \eqref{eq:JJ}; the one-time law follows because at fixed $t$ the pair $(v_1,v_2)$ is a pair of independent circular Gaussians, as in the free case. The statements about line spectra follow from the form of $S_1S_2$.
\end{proof}

Stationary Gaussian velocities are what a linear coupling to a continuum of unresolved modes produces, given a Gaussian preparation that is stationary under the coupled flow; whether such a preparation exists depends on the coupled dynamics being stable. Appendix~\ref{app:bathmodel} gives the simplest model in which it does: each coefficient velocity is treated as the velocity of a free coordinate with a positive kinetic term, coupled bilinearly to its own bath of oscillators with a counter-term that keeps it free, in the form of Caldeira and Leggett \cite{caldeira}, and prepared in the stationary Gaussian ensemble of the coupled system---the coefficient-level analogue of the canonical ensemble in trace dynamics \cite{adler} and of the tracial state in Appendix~\ref{app:bath}.

\begin{example}[Ohmic bath, Markov limit]\label{ex:ohmic}
In that model with an Ohmic bath and large cutoff each velocity is an Ornstein--Uhlenbeck process, $C_r(t)=\sigma_r^2e^{-\gamma_r|t|}$ with friction rate $\gamma_r$ and $S_r(\omega)=2\sigma_r^2\gamma_r/(\gamma_r^2+\omega^2)$. Then, with $\Gamma=\gamma_1+\gamma_2$,
\begin{equation}
\langle J(t)J(0)\rangle=2\sigma_1^2\sigma_2^2e^{-\Gamma|t|},\qquad
D_J=\frac{4\sigma_1^2\sigma_2^2}{\Gamma},\qquad
\operatorname{Var}Y(\tau)=D_J\Bigl[\tau-\frac{1-e^{-\Gamma\tau}}{\Gamma}\Bigr].
\label{eq:ohmic}
\end{equation}
For $\Gamma\tau\ll1$ this is $2\sigma_1^2\sigma_2^2\tau^2$, the free result; for $\Gamma\tau\gg1$ it is $D_J\tau$. The bath does not change what the coefficient is at any one time; it changes how it moves, turning a frozen random variable into a mixing process whose time integral is asymptotically Gaussian with variance $D_J\tau$, by the central limit theorem for quadratic functionals of Gaussian processes \cite{breuermajor}. The friction $\gamma_r$ and the noise strength are tied by the stationary preparation through the classical fluctuation--dissipation relation (Appendix~\ref{app:bathmodel}); this is a relation between the velocity noise and the velocity friction, and it is not the state-space noise--drift relation of Section~\ref{sec:normdrift}, which the projection has still to produce.
\end{example}

\begin{figure}[!htbp]
\centering
\includegraphics[width=\textwidth]{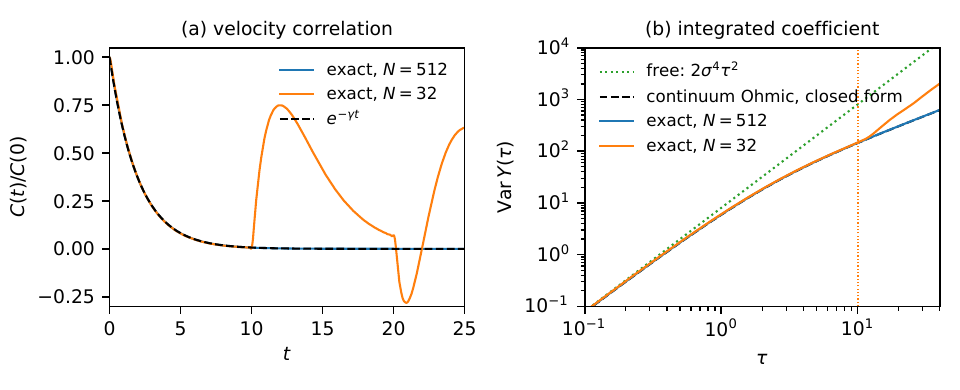}
\caption{The linearised-bath model of Appendix~\ref{app:bathmodel} with $\gamma=1/2$, unit mass and unit preparation temperature. (a)~Exact velocity correlation of the discretised Ohmic bath for $N=512$ modes (cutoff $\omega_c=60$) and $N=32$ ($\omega_c=20$), against the Markov-limit $e^{-\gamma t}$; the small bath recurs at $2\pi N/\omega_c\approx10$. (b)~Variance of the integrated coefficient $Y(\tau)$: the free ballistic law, the closed form \eqref{eq:ohmic}, and the exact finite-bath results, which follow the closed form until the recurrence time (vertical dotted line) and then return to ballistic growth.}\label{fig:bath}
\end{figure}

Figure~\ref{fig:bath} compares the closed form with the exact finite-bath dynamics: for $N=512$ modes the integrated variance agrees with \eqref{eq:ohmic} to within $1\%$ out to $\tau=8$, and a bath of $32$ modes follows it until its recurrence time and then returns to ballistic growth, as the fermionic reservoir of Appendix~\ref{app:bath} recurs on the scale $2\pi\hbar N/E_W$. In this model the finite-bath growth is always ballistic in the end, because the counter-term makes the total momentum of block and bath a conserved quantity, which leaves a zero-frequency line of weight $T/M_{\rm eff}$ in the velocity spectrum at every finite $N$ (Appendix~\ref{app:bathmodel}); the continuum limit removes it. Diffusion requires the continuum, in both calculations.

\begin{remark}[A bath that drives both coefficient directions identically produces no noise]\label{rem:common}
If the velocities are jointly circular but correlated, with $C_{rs}(t)=\langle v_r(t)v_s^*(0)\rangle$, the general formula \eqref{eq:JJ_general} gives
$\langle J\rangle=-2\operatorname{Im}C_{12}(0)$ and
$\operatorname{Cov}(J(t),J(0))=2\operatorname{Re}[C_{11}C_{22}^*-C_{12}C_{21}^*](t)$.
For $v_2=v_1$ the covariance vanishes identically, since $J=2\operatorname{Im}|v_1|^2=0$: a bath that drives the two coefficient directions identically produces no coefficient noise, the microscopic counterpart of the statement in Section~\ref{sec:solvay} that a channel coupling identically to every alternative cannot select. A phase-rotated drive $v_2=e^{i\zeta}v_1$ instead gives $J=2\sin\zeta\,|v_1|^2$, with mean $2\sigma_1^2\sin\zeta$ and covariance $4\sin^2\zeta\,|C_{11}(t)|^2$: a coherent response together with noise of white-noise weight $4\sin^2\zeta\int|C_{11}|^2\dd t$.
\end{remark}

\paragraph{A sea from the programme.}
The Ohmic bath of Example~\ref{ex:ohmic} is a stand-in. The programme has, however, one explicit interacting sector with a computed collective spectrum: the relational master-matrix extension of GTD \cite{master}, whose Gram sector reduces to a spin--Sutherland Hamiltonian \cite{sutherland} and whose stable rank-$r$ textures place $M=N/r$ equally spaced sites in each of $r$ orthogonal colours, with a position Hessian of Fourier eigenvalues $4\ell^2(M-\ell)^2$ in natural units \cite[Sec.~6]{projectvii}. The audit of Ref.~\cite{projectvii} notes that this interaction stabilises textures but that ``Hamiltonian stability is not damping'', and its Appendix~C, written for the free master matrix, lists a size-dependent frequency distribution among the three premises a broadening mechanism must violate; what it does not do is compute the spectrum the texture offers to a block attached to it. That spectrum is dispersive, and the computation is short.

\begin{example}[The Sutherland texture as a sea]\label{ex:sutherland}
One colour chain of the texture is $M$ particles of mass $m$ on a circle of length $\mathcal L$ with potential $V=\mathsf g(\pi/\mathcal L)^2\sum_{i<j}\csc^2[\pi(x_i-x_j)/\mathcal L]$, $\mathsf g=s^2/m$ in the notation of Ref.~\cite{projectvii}, linearised about $x_n=n\mathcal L/M$. The Hessian eigenvalues $4\mathsf g(\pi/\mathcal L)^4\ell^2(M-\ell)^2$ give the normal-mode frequencies
\begin{equation}
\omega_\ell=\frac{2\pi^2\sqrt{\mathsf g/m}}{\mathcal L^2}\,\ell(M-\ell),\qquad \ell=1,\ldots,M-1,
\label{eq:sutherland_modes}
\end{equation}
with three properties. All $\omega_\ell$ are integer multiples of $\omega_1'=2\pi^2\sqrt{\mathsf g/m}/\mathcal L^2$, so the linearised chain is exactly periodic with period $T_{\rm rec}=\mathcal L^2/(\pi\sqrt{\mathsf g/m})$, in the spirit of the all-size recurrence that Ref.~\cite{projectvii} proves for the free master matrix, but with a dispersive spectrum. The band is gapless, with $\omega\simeq vk$ at long wavelength, $v=\pi\sqrt{\mathsf g/m}\,M/\mathcal L$, and edge $\Omega_{\max}=\pi^2\sqrt{\mathsf g/m}\,M^2/(2\mathcal L^2)$ at $\ell=M/2$. The low-frequency mode spacing is $4\Omega_{\max}/M$ and the period is $T_{\rm rec}=\pi M^2/(2\Omega_{\max})$. How these scale with $M$ depends on what is held fixed, and that is a choice the programme has not made: the texture of Ref.~\cite{projectvii} is written on a circle of fixed length, for which $\Omega_{\max}\propto M^2$, the period is independent of $M$ and no continuum forms; we take instead the thermodynamic limit at fixed spacing $a=\mathcal L/M$ and fixed $\mathsf g$, $m$ and coupling, for which $\Omega_{\max}=\pi^2\sqrt{\mathsf g/m}/(2a^2)$ and $v=\pi\sqrt{\mathsf g/m}/a$ are fixed, the spacing $4\Omega_{\max}/M$ closes, and the period grows as $M^2$. Ref.~\cite[Sec.~9.2]{microclosure} asks that exactly this scaling be stated. Attach one real block coordinate $q$ of mass $m_b$, free as in Appendix~\ref{app:bathmodel}, to the local strain $e=x_1-x_0$ through $-cqe$ with the counter-term $\frac12c^2\chi_eq^2$ that keeps it free, $\chi_e$ the static strain susceptibility; the four real components of $(v_1,v_2)$ are attached to four independent chains (a common bond would trigger the cancellation of Remark~\ref{rem:common}, and distant bonds of one chain are not independent, since one-dimensional phonon correlations arrive delayed rather than attenuated). In normal modes this is the Caldeira--Leggett model with $\Omega_k\to\omega_\ell$ and $c_k\to ce_\ell/\sqrt m$, $e_\ell$ the strain amplitude of real mode $\ell$, for which the two real modes at frequency $\omega_\ell$, $\ell<M/2$, together carry strain weight $(8/M)\sin^2(\pi\ell/M)$, and the single mode at $\ell=M/2$ carries $4/M$. The spectral density is Ohmic at low frequency,
\begin{equation}
\varrho(\omega)=\gamma m_b\,\omega\,[1+O(\omega/\Omega_{\max})],\qquad
\gamma=\frac{\pi^3c^2}{16\,m\,m_b\,\Omega_{\max}^3},
\label{eq:sutherland_ohmic}
\end{equation}
rising to a van Hove singularity at $\Omega_{\max}$ (Figure~\ref{fig:sutherland}a); the block velocity is Ornstein--Uhlenbeck with rate $\gamma$ in the limit $4\Omega_{\max}/M\ll\gamma\ll\Omega_{\max}$, and Example~\ref{ex:ohmic} applies with $\Gamma=2\gamma$. The conserved total momentum of the Caldeira--Leggett form leaves the zero-frequency line $T/M_{\rm eff}$ of Appendix~\ref{app:bathmodel} with
\begin{equation}
M_{\rm eff}=m_b+\sum_{\text{real modes}}\frac{c^2e_\ell^2}{m\,\omega_\ell^4}\simeq m_b\Bigl[1+\frac{\pi\gamma M}{12\,\Omega_{\max}}\Bigr],
\label{eq:sutherland_meff}
\end{equation}
the analogue of $1+\pi\gamma N/\omega_c$ for the uniform discretisation, so the line vanishes as $12\Omega_{\max}/(\pi\gamma M)$. With $\Omega_{\max}=40$, $\gamma=\frac12$ and $T=m=m_b=1$: at $M=4096$ the mode spacing is $0.078\gamma$, the exact velocity correlation lies within $0.034$ of $e^{-\gamma t}$ for $t\leq12$, $\operatorname{Var}Y(\tau)$ is within $5\%$ of the closed form \eqref{eq:ohmic} out to $\tau=8$ (the residual is the finite-cutoff correction: at $\Omega_{\max}/\gamma=160$ with the same mode spacing it is $3\%$ at $\tau=1$ and $1\%$ at $\tau=8$), the line is $T/M_{\rm eff}=0.069$, and the period is $6.6\times10^5$; at $M=1024$ the spacing is $0.31\gamma$ and the line $0.23$; at $M=256$ the spacing is $1.25\gamma$ and the correlation oscillates about $T/M_{\rm eff}=0.53$; at $M=24$ the spacing is $13\gamma$, the correlation never falls below $0.82$ and averages $T/M_{\rm eff}=0.90$, and the variance grows ballistically throughout (Figure~\ref{fig:sutherland}c; ancillary script). The sea and its spectrum are the programme's; the scaling with $M$, the coupling $c$, its strength and its attachment to the strain are inserted.
\end{example}

\begin{figure}[!htbp]
\centering
\includegraphics[width=\textwidth]{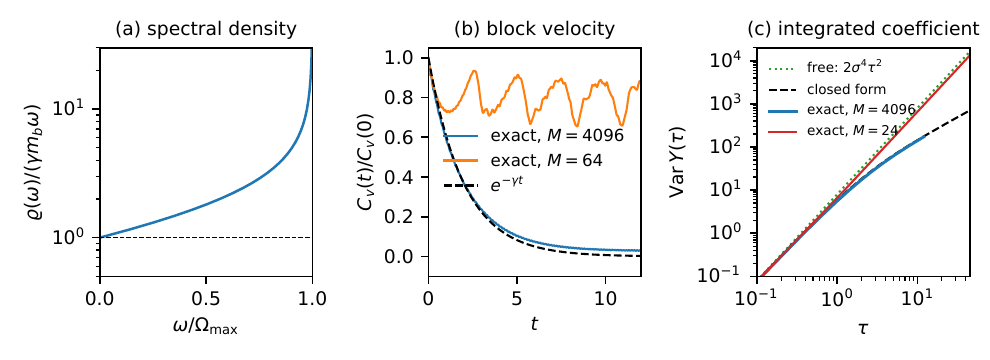}
\caption{A block attached to the Sutherland texture (Example~\ref{ex:sutherland}; $\Omega_{\max}=40$, $\gamma=\frac12$, unit masses and temperature). (a)~Spectral density of the strain coupling in units of the Ohmic form $\gamma m_b\omega$, from the discrete modes at $M=2048$: Ohmic at low frequency, a van Hove singularity at the band edge. (b)~Exact velocity correlation of the block for $M=4096$ and $M=64$ against $e^{-\gamma t}$; the small chain has its lowest mode at $4.9\gamma$ and never decorrelates. (c)~Variance of the integrated coefficient: the free ballistic law, the closed form \eqref{eq:ohmic}, the exact $M=4096$ result, and the exact $M=24$ result, which stays ballistic.}\label{fig:sutherland}
\end{figure}

Two features of the sea deserve separate statement. The coupling has to be local and invariant under a uniform translation of the chain for the Ohmic form: a coupling to a single displacement, whose mode weight is $2/M$ without the factor $\sin^2(\pi\ell/M)$, gives $\varrho(\omega)\simeq\pi c^2/(4m\Omega_{\max}\omega)$, the infrared divergence of one-dimensional displacement fluctuations, and no friction kernel. And the coupling has to act on the block's coordinate. A purely kinetic cross-coupling $k_v\,\dot q\,\dot e$---the velocity--velocity structure of the cross terms $\eta\beta_2\Tr[(D_-q_B^\dagger)(D_+q_F)]$ of the action \eqref{eq:action}, now between different atoms---leaves the Lagrangian independent of $q$, so $p_q=m_b\dot q+k_v\dot e$ is conserved exactly; the block velocity $\dot q=(p_q-k_v\dot e)/m_b$ keeps a non-decaying component, its correlation has zero-frequency weight exactly $T/m_b$, and $\operatorname{Var}Y(\tau)\to8(T/m_b)^2\tau^2$: the free coefficient stays frozen under any amount of kinetic mixing (ancillary script, $M=256$, $k_v=0.3$). This is the coefficient-level form of the audits' findings that the additive GTD action carries no exchange current and that the free master matrix recurs \cite{projectvii,audit}: an interaction that is to unfreeze the coefficient must be a coordinate-dependent vertex, of the type of \eqref{eq:vertex}, not a kinetic bilinear.

\paragraph{Which block the model describes.}
Proposition~\ref{prop:bath} is a statement about the coefficient given stationary Gaussian velocities, and it holds for whatever dynamics produces them. The model of Appendix~\ref{app:bathmodel} and Example~\ref{ex:sutherland} produce them by giving the block the positive kinetic term $\frac12(|\dot Q_1|^2+|\dot Q_2|^2)$. The Grassmann block of \eqref{eq:free_J} does not have that term: its kinetic form is the coefficient $J$ itself, of signature $(2,2)$, and one may ask whether a sea could be coupled to such a block without a growing mode. The question is not the right one, for the following reason.

\begin{lemma}[Positive states do not see the Grassmann coefficient]\label{lem:positive}
Let $\mathcal A$ be the exterior algebra generated by the dynamical generators $\theta_1,\theta_2$ of \eqref{eq:free_J}, with $\theta_r^\dagger=\theta_r$, $\theta_r^2=0$, $\theta_1\theta_2=-\theta_2\theta_1$. Every positive normalised linear functional $\omega$ on $\mathcal A$ satisfies $\omega(\theta_1\theta_2)=0$, hence $\omega(V_F^\dagger V_F)=iJ\,\omega(\theta_1\theta_2)=0$ for every value of the coefficient $J$. On the Clifford reconstruction $\theta_r^2=\frac12$, every positive normalised functional has $\omega(i\theta_1\theta_2)=-\varsigma/2$ with $\varsigma\in[-1,1]$, and
\begin{equation}
\omega(V_F^\dagger V_F)=\tfrac12(|v_1|^2+|v_2|^2)-\varsigma\operatorname{Im}(v_1^*v_2)\;\geq\;\tfrac12(1-|\varsigma|)(|v_1|^2+|v_2|^2),
\label{eq:clifford_kinetic}
\end{equation}
positive semidefinite for every state and definite unless $|\varsigma|=1$.
\end{lemma}
\begin{proof}
With $x=\theta_1\theta_2$ one has $x^\dagger=\theta_2\theta_1=-x$ and $x^\dagger x=-x^2=0$, so Cauchy--Schwarz gives $|\omega(x)|^2\leq\omega(1)\omega(x^\dagger x)=0$; this is the argument of Appendix~\ref{app:micro} for $b=\beta_1\beta_2$, applied to the dynamical generators, and it is a case of the theorem of Ref.~\cite[Thm.~4.1]{physrecon} that a positive normalised functional on a finite-dimensional star algebra factors through the quotient by any star-stable nilpotent ideal. On the Clifford reconstruction $\theta_1=\sigma_1/\sqrt2$, $\theta_2=\sigma_2/\sqrt2$ and $i\theta_1\theta_2=-\sigma_3/2$, whose expectation in a state of the $2\times2$ matrix algebra is $-\varsigma/2$ with $\varsigma=\omega(\sigma_3)$; the bound follows from $|\operatorname{Im}(v_1^*v_2)|\leq\frac12(|v_1|^2+|v_2|^2)$.
\end{proof}

So the block whose kinetic form is $J$ admits no state in which $J$ is visible, and a stationary preparation for it is not a question that arises: a positive state space for the dynamical generators is a reconstruction with a changed product, exactly as it is for the auxiliary ones in Section~\ref{sec:projection_result}, and Ref.~\cite{physrecon} reaches the same conclusion for the marked Grassmann realisation of the master matrix (``a different physical product is therefore indispensable''). The mixed-sign instability of a $(2,2)$ kinetic form coupled to a positive bath---a negative-kinetic direction acquires negative friction and a growing mode at a rate close to $\gamma$ (ancillary script)---remains a true statement about such forms; by the lemma, no positive state realises one, so it does not arise. On the Clifford reconstruction the kinetic form \eqref{eq:clifford_kinetic} is positive semidefinite in every state and definite for $|\varsigma|<1$; in the tracial state $\varsigma=0$ it is $\frac12(|v_1|^2+|v_2|^2)$, the kinetic term of the model with unit mass and $c_0\eta^3=1$, with $J$ the coefficient of the traceless part. A bilinear coupling of that block to a positive sea has the Caldeira--Leggett structure, and the model of this section describes the reconstructed block in its tracial state---subject to two provisos. In a sector state $\varsigma\neq0$ the kinetic form is $\frac12\mathbf v^T(\id-\varsigma\mathsf M)\mathbf v$, the stationary velocity covariance is $T(\id+\varsigma\mathsf M)/(1-\varsigma^2)$, and the coefficient acquires the mean $\langle J\rangle=4\varsigma T/(1-\varsigma^2)$: a coherent gain of the kind Section~\ref{sec:failures} warns against, so causal consistency excludes a non-tracial sector state unless that gain is attached to the identity. And the Lagrangian term \eqref{eq:free_aux} is $-c_0\eta^3b\Tr(V_F^\dagger V_F)$ with $b^\dagger=-b$, so in any positive state $\omega(b)$ is imaginary: the real kinetic energy that the block dynamics needs is the equilibrium reading $b\mapsto-1$ of Ref.~\cite{book}, which is not the value of $b$ in any positive state, while the coefficient whose history drives the ASA term uses the sector value $b\mapsto s_\beta i\rho_0$ of Section~\ref{sec:projection_result}. The two readings are mutually exclusive as values of one state; the monograph obtains the first by its scalar projection, a functional that is not positive, and the model of this section assumes a real kinetic term for the block from that or another source. The lemma removes the sign obstruction on the side of the dynamical generators; it does not remove the need for that input. What the reconstruction does not supply is which positive state, and in particular the polarisation $\varsigma$ and the sign $s_\beta$; the conditional repair of Ref.~\cite[Sec.~4.1]{physrecon}, a one-parameter family of positive product states, is a freedom of the same kind. The result of this subsection is therefore the following. The ingredient whose absence the free calculation exposed---an integrable correlation of the coefficient with a finite white-noise weight---is produced, with a computable strength and an accompanying friction, as soon as the coefficient velocities are mixed by a continuum of unresolved modes through a coordinate-dependent coupling, and the programme's own texture, taken long at fixed spacing, supplies such a continuum with the dispersion, criterion and recurrence of Example~\ref{ex:sutherland}. A white-noise weight is necessary for a law of the form \eqref{eq:SSE} and not sufficient: the map from the coefficient to a projected anti-self-adjoint operator, the counter-term and the response of Section~\ref{sec:history}, with the preparation probabilities of the effective state, is what Section~\ref{sec:closure} asks for, and the next subsection states what the sea's response must satisfy.

\subsection{What the sea's response must satisfy: an Einstein relation}\label{sec:einstein}
Proposition~\ref{prop:history} asks two things of a derived history: that the linear law carry the counter-term with $\varkappa=1$ where the records can see it, and that the history lean towards the leading record with gain $2\mathsf D$. If the history is the integrated coefficient of a sea, the lean is the sea's response to the system, and a sea in equilibrium does not respond arbitrarily: its response is tied to its fluctuations. The following proposition makes the consequence for the gain explicit. It assumes the mean-field form of the back-reaction, in which the sea sees the system through a record expectation; that assumption is the nonlinear ingredient, and the proposition says what its gain must be, not that the mechanism is derived. It says nothing about the counter-term, and we return to that afterwards.

\begin{proposition}[Einstein relation for the response gain]\label{prop:einstein}
Let the history be $Y=g\int J\,\dd\tau$, with $J$ a phase-space function of a classical sea in canonical equilibrium at temperature $T$, of stationary autocovariance $C_J(t)$ with $C_J(\infty)=0$, white-noise weight $D_J=\int_{-\infty}^\infty C_J\,\dd t$ and correlation time $\tau_c:=D_J/[2C_J(0)]$, so that in the white-noise limit $\mathsf D=g^2D_J$ for one channel. Let the system act on the sea through the perturbation $-\lambda_{\rm b}\,x(\tau)\,J$ of the sea's Hamiltonian, with $x=\langle L\rangle_\psi$ and the same phase-space function $J$. Then, to linear order in $\lambda_{\rm b}$, the mean of the coefficient along a given system history is
\begin{equation}
\langle J\rangle(\tau)=\lambda_{\rm b}\int_0^\tau\chi_J(\tau-\tau')\,x(\tau')\,\dd\tau',\qquad
\chi_J(t)=-\frac{\theta(t)}{T}\,\dot C_J(t),
\label{eq:kubo}
\end{equation}
and when $x$ varies slowly on the scale $\tau_c$ the history acquires, in its white-noise limit, the drift $\mathsf b=g\lambda_{\rm b}C_J(0)\,x/T$. Comparing with the gain $2\alpha\mathsf D\bar\ell$ of Proposition~\ref{prop:history},
\begin{equation}
\alpha=\frac{g\lambda_{\rm b}C_J(0)/T}{2g^2D_J}=\frac{\lambda_{\rm b}}{4gT\tau_c},
\qquad\text{so that the Born gain requires}\qquad \lambda_{\rm b}=4gT\tau_c .
\label{eq:einstein}
\end{equation}
For the Ohmic sea of Example~\ref{ex:ohmic}, $C_J(t)=2\sigma_1^2\sigma_2^2e^{-\Gamma|t|}$, $\tau_c=1/\Gamma$ and the condition is $\lambda_{\rm b}\Gamma=4gT$. Outside the slow-variation regime the drift carries the memory kernel $\chi_J$ and the normalised law is non-Markovian; the Markov law of Section~\ref{sec:joint} is the limit $\tau_c\ll1/(2\delta^T\mathsf D\delta)$, $\delta$ an eigenvalue difference.
\end{proposition}
\begin{proof}
Equation~\eqref{eq:kubo} is the classical Kubo formula for the response of a phase-space observable to a perturbation $-\lambda_{\rm b}xJ$ of the Hamiltonian in canonical equilibrium \cite{kubo}; its zero-frequency value is $\int_0^\infty\chi_J=[C_J(0)-C_J(\infty)]/T=C_J(0)/T$. The rest is comparison of the gain with Proposition~\ref{prop:history} at $\mathsf D=g^2D_J$ and $D_J=2\tau_cC_J(0)$; the counter-term coefficient is not touched.
\end{proof}

Four consequences follow, and one caution. First, the response gain is not a free parameter of a sea in equilibrium: $\alpha$ is fixed by the coupling, the temperature and the correlation time, and the Born gain holds on the codimension-one locus \eqref{eq:einstein}. Second, by Proposition~\ref{prop:history}(iv) and Section~\ref{sec:failures}, a sea off that locus produces records that signal; the causal consistency requirement of this paper, applied to a sea that responds to the system, therefore includes the requirement that its response gain and its noise strength stand in the Einstein ratio \eqref{eq:einstein}. Third, the sign matters: a response of the wrong sign relative to the noise channel gives $\alpha<0$, for which the log-odds of Proposition~\ref{prop:history}(iii) are recurrent. The sign is not innocent for the GTD coefficient, which is built from velocities. The Kubo formula refers to the phase-space function that appears in the perturbation; if the perturbation is $-\lambda_{\rm b}xJ(p)$ in the Hamiltonian while the history is driven by the velocity bilinear $J(v)$, the perturbation changes the relation between momenta and velocities and the static response of $J(v)$ acquires a contact term that reverses its sign, $\langle J(v)\rangle=-\lambda_{\rm b}xC_J(0)/T$ against $+\lambda_{\rm b}xC_J(0)/T$ for $J(p)$ (ancillary script, exact Gaussian evaluation); the magnitude of \eqref{eq:einstein} is unchanged, and which sign of $\lambda_{\rm b}$ gives $\alpha>0$ is part of the specification of the coupling. Fourth, the relation is between the gain of the sea's response to a c-number and its fluctuations; a sea that responds to the operator $L$, and is then averaged over, gives the Lindblad equation \eqref{eq:ensemble} and no selection (Appendix~\ref{app:bath}), so the mean-field form of the back-reaction is where the nonlinearity of a collapse law enters, and \eqref{eq:einstein} is the one value of its gain compatible with no-signalling once the counter-term is supplied.

The caution concerns the counter-term, which the Einstein relation does not touch. A history produced by a smooth coefficient drives the linear law $\dot\phi=gJL\phi$, whose white-noise limit has $\varkappa=-1$ by Proposition~\ref{prop:history}(ii), whatever the gain. On the locus \eqref{eq:einstein} such a law is Born only where the counter-term is invisible by Proposition~\ref{prop:history}(ii)---two records with eigenvalues symmetric about zero, or records on which $\ell_j^T\mathsf D\ell_j$ is constant, which includes the local detector channels of Section~\ref{sec:solvay} with equal diffusion; for a projector $L$ with eigenvalues $1,0$ it selects the eigenvalue-one record with probability one from any initial weights, and for three records at $\ell=1,0,-1$ it never selects the middle one, whereas adding the term $-g^2D_JL^2\dd\tau$ (Stratonovich) to the linear law restores Born to within $0.03$ in both cases, the residual being the finite-$\Gamma_s$ correction (ancillary script, Ornstein--Uhlenbeck sea with relaxation rate $\Gamma_s=20$ at selection rate $2\mathsf D\delta^2=1$; for the symmetric two-record case the exit frequencies on the locus are Born to within $0.003$, and for $\lambda_{\rm b}\Gamma_s/(4gT)=\frac12$ and $2$ they follow $I_{p_0}(\alpha,\alpha)$). The counter-term is therefore a second, independent demand on the projected generator whenever the record operator has three or more collinear eigenvalues, as the collective operator of Example~\ref{ex:amplification} does: besides the fluctuating term $gJL$ it must contain the deterministic term $-g^2D_JL^2$, quadratic in the record operator and with a coefficient fixed by the noise strength. For the local detector channels of the Solvay model it is invisible, and there the gain alone decides. In the quantum derivation of the CSL law this term is the anticommutator part of the open-system generator, the sea's second-order response to the operator; a mean-field sea does not produce it, and a GTD projection has to. Whether the GTD coarse-graining places the sea on the locus \eqref{eq:einstein} and supplies the counter-term is the content of the response calculation that Section~\ref{sec:closure} asks for; Refs.~\cite{projectvii,physrecon} observe that an equilibrium fluctuation--dissipation relation does not by itself produce the nonlinear response fraction, and the present section states exactly what else is needed: the value of $\lambda_{\rm b}$ relative to $gT\tau_c$, and the quadratic term. If the coefficient noise retains its finite correlation time at the level of the effective law, that law is a coloured-noise collapse model \cite{diosigisinstrunz,adlerbassi,bassiferialdi,tilloywiseman}, and the Markov limit assumed in Section~\ref{sec:joint} must be justified by comparing $\tau_c$ with the induced collapse time.

\subsection{What the projection must do}\label{sec:projection_result}
The coefficient \eqref{eq:free_aux} carries the auxiliary factor $b=\beta_1\beta_2$, and the monograph uses two different proposals for what becomes of it: the equilibrium functional takes $b\mapsto-1$ and gives a real coefficient Hamiltonian, while the collapse discussion proposes $b\mapsto s_\beta i\rho_0$, with $s_\beta$ a sign and $\rho_0$ a positive constant, and a non-Hermitian effective equation \cite[Secs.~7.3.6 and 9.2.1]{book}. The most direct third possibility, representing the auxiliary product and adjoint unchanged on a positive Hilbert space, gives zero: a star representation $\pi$ preserves products and adjoints, and since $b^\dagger=-b$ and $b^2=0$ imply $b^\dagger b=0$, one has $\|\pi(b)u\|^2=\langle u,\pi(b^\dagger b)u\rangle=0$ for every $u$, so $\pi(b)=0$ (Appendix~\ref{app:micro}). This is the same reason quantisation uses Clifford rather than Grassmann generators. The image of \eqref{eq:free_aux} is then zero whatever the preparation and whatever operator coefficient is appended; for two stable records this route gives $p(\tau)=p_0$ and a vanishing ASA history, an identity operation that cannot signal and cannot select. The random coefficient history of Proposition~\ref{prop:bath} must not be confused with this zero physical history. A nonzero replacement such as $b\mapsto i\rho_0$ is not a positive expectation on the auxiliary algebra, by Cauchy--Schwarz; the monograph calls it a nonhomomorphic coarse-graining. It can, however, be read as a homomorphism of a different algebra followed by a sector choice: Clifford generators with $\beta_r^2=\frac12$ give $b^\dagger=-b$ and $b^2=-\frac14$, so $b$ has the two eigenvalues $\pm\frac i2$, which is $b\mapsto s_\beta i\rho_0$ with $\rho_0=\frac12$. That reading has to be given on the interacting observables together with the preparation probabilities and the effective evolution it induces, and Lemma~\ref{lem:positive} shows that the same change of product for the dynamical generators is forced, not merely presupposed, by the existence of any positive state in which the coefficient is visible; Ref.~\cite[Thm.~4.1]{physrecon} states the general theorem and adds the caution that it does not apply to a reconstructed product merely because Grassmann symbols were used to calculate it. The reading also carries a warning: with Clifford generators the velocity bilinear acquires the scalar part $\frac12(|v_1|^2+|v_2|^2)$, a coefficient of nonzero mean, and whether the projection attaches it to the identity, where it is harmless, or to a collective operator, where it is a coherent gain of the kind that signals in Section~\ref{sec:failures}, is part of what must be specified. What the sector sign $s_\beta$ does and does not affect can be said now. The centred noise of the collapse law is insensitive to it, since $W$ and $-W$ have the same law; it enters the sign of that coherent gain, and it enters the relative sign of the response and the noise in Proposition~\ref{prop:einstein}---where the wrong sign gives $\alpha<0$ and a recurrent law---provided the back-reaction coupling $\lambda_{\rm b}$ does not itself carry $s_\beta$, in which case the relative sign is unchanged. We have not found a derivation relating $s_\beta$ to the orientation of the Lorentzian leaf of Section~\ref{sec:splitform}; the two indices of opposite sign that Ref.~\cite{projectiii} reports on the two $\omega$-conjugate orientations are a hint in that direction and no more.

Two further diagnostics bound the possibilities. A self-adjoint kinetic completion, in which a positive boson-dependent matrix $S(q)$ normalises the fermionic kinetic term and induces the Hermitian connection $A_i=\frac i2[S^{-1},\partial_iS]$, gives a self-adjoint Hamiltonian and hence no ASA part at all, however entangling it is. And the positive fermionic reservoir of Appendix~\ref{app:bath}, coupled through a Hermitian current to a conserved record observable $Q$, supplies noise with a finite white-noise weight $D_B=\pi\hbar v^2/(2E_W)$ in the continuum limit and dephases coherences at the rate $\Gamma_{\rm d}=\lambda^2D_B/\hbar^2$, but, because $[P_j\otimes\id,H]=0$, it preserves every record-sector weight and selects nothing; in its tracial preparation the bath carries no distinguishable final state for different records either. That is decoherence \cite{zurek}, and it illustrates the random-phase comparison of Section~\ref{sec:normdrift}: both the random-phase law and the CSL law with $L=\sqrt{\Gamma_{\rm d}}\,Q$ reproduce $\dot\rho=-\frac{\Gamma_{\rm d}}2[Q,[Q,\rho]]$, and only the second changes weights.

\section{The calculation that remains}\label{sec:closure}
\begin{table}[!htbp]\small\caption{What has been established, and where.}\label{tab:status}
\begin{tabularx}{\textwidth}{@{}p{.30\textwidth}X@{}}
\toprule Question & Answer in this paper\\\midrule
Does the $(3,3)$ form give a causal order? & It gives a null cone and $SO(3,3)$ invariance, not the two-component timelike structure of a future--past order (Section~\ref{sec:splitform}).\\[3pt]
Does collapse in Connes time remove the hypersurface problem? & It defers it: the question is not posed at the level at which selection happens, and returns as the requirements that emergent records be non-signalling and that the emergent description be covariant (Sections~\ref{sec:connes}, \ref{sec:continuum}).\\[3pt]
Does the effective law select one record with Born frequencies? & Yes, asymptotically, for distinct stable signatures and sustained exposure, with the instrument trace rule as premise (Section~\ref{sec:continuous}).\\[3pt]
Does norm conservation fix the collapse drift? & No. The bare norm-preserving law is solved exactly and never selects; the Born-restoring response is tangent to the sphere (Proposition~\ref{prop:bare}).\\[3pt]
What must a derived history satisfy? & Its conditional drift must be $2\mathsf D\bar\ell$ and, for general record operators, its linear law must carry the counter-term with unit coefficient; a gain $\alpha\in(0,1)\cup(1,\infty)$ completes with the wrong frequencies and signals. For an equilibrium sea the gain is fixed by fluctuation--dissipation, and the Born gain is the Einstein relation $\lambda_{\rm b}=4gT\tau_c$; the counter-term is a separate demand (Propositions~\ref{prop:history}, \ref{prop:einstein}).\\[3pt]
Can a split-signature geodesic length be the norm? & Not on a linear space of states without selecting a maximal positive subspace, which is the same kind of selection that picks a Lorentzian sector (Section~\ref{sec:geodesicnorm}).\\[3pt]
Are Solvay and Bell records non-signalling? & Yes, by the standard theorem, for local complete instruments on an assumed emergent laboratory structure, with Bell correlations as inputs (Section~\ref{sec:nosig}).\\[3pt]
Where may entanglement enhancement enter? & In the unravelling, without signalling but without ensemble consequences and only as a suppression below the operator-fixed rate; or in the spectrum of fixed collective operators, observably. Not as a scalar rate on the generator (Propositions~\ref{prop:ent}, \ref{prop:bound}, Example~\ref{ex:amplification}).\\[3pt]
Does the GTD coefficient fluctuate? & Not in the free theory, and not under any kinetic cross-coupling, which conserves the block momentum. Once its velocities are mixed by a continuum through a coordinate-dependent coupling it is a mixing process with white-noise weight $D_J=\pi^{-1}\!\int\Tr[\mathsf M\mathsf S\mathsf M\mathsf S]$; the Sutherland texture of the master matrix, taken long at fixed spacing, is such a continuum, with Ohmic friction, criterion $4\Omega_{\max}/M\ll\gamma$ and period $\pi M^2/(2\Omega_{\max})$ (Proposition~\ref{prop:bath}, Example~\ref{ex:sutherland}).\\[3pt]
Is the Grassmann block the physical block? & No. Positive states annihilate the bilinear whose coefficient is $J$; on the Clifford reconstruction the kinetic form is positive semidefinite in every state, and the bath model describes the reconstructed block in its tracial state, given a real kinetic term that no positive value of the auxiliary factor supplies (Lemma~\ref{lem:positive}).\\[3pt]
Is a nonzero ASA collapse channel derived? & No. An unchanged positive representation sends the coefficient to zero; a Clifford treatment of the generators realises the proposed replacement $b\mapsto s_\beta i\rho_0$ with $\rho_0=\frac12$ as a homomorphism plus sector choice, but the interacting projection, the counter-term, the response gain and the joint law \eqref{eq:push} remain to be constructed (Section~\ref{sec:projection_result}).\\\bottomrule
\end{tabularx}
\end{table}
Let $X$ denote regulated microscopic GTD variables and $\Phi_\tau X$ their deterministic interacting evolution. A preparation measure $\mu_\psi$ describes the unresolved microscopic conditions consistent with an effective initial state $\psi$; $\mathcal R$ maps a microscopic trajectory to its effective state history and $\mathcal Y$ to its effective fluctuation history. For sets $A,E$ of histories the joint physical probability is
\begin{equation}
\mathbb P_{{\rm GTD},\psi}(A,E)=\int\mathbf1_{\{\mathcal R[\Phi_\cdot X]\in A\}}\mathbf1_{\{\mathcal Y[\Phi_\cdot X]\in E\}}\,\mu_\psi(\dd X).
\label{eq:push}
\end{equation}
In words: count the microscopic preparations whose trajectories give both the specified state history and the specified fluctuation history. Everything in this paper is a constraint on \eqref{eq:push}. To reproduce the effective law it must give \eqref{eq:mix} and \eqref{eq:update} together, or the joint generator \eqref{eq:generator}, including the drift of $Y$ correlated with the changes of $p$, which no single-force covariance supplies; a coarse-graining that merely conserves the norm fails it (Proposition~\ref{prop:bare}); its history must lean towards the leading record with gain $2\mathsf D$, which for a sea in equilibrium is the Einstein relation \eqref{eq:einstein} (Propositions~\ref{prop:history}, \ref{prop:einstein}), and its linear law must carry the counter-term with unit coefficient, which a coefficient drive does not supply; state dependence may sit only in the unravelling or in the operators (Propositions~\ref{prop:ent}, \ref{prop:bound}); and a split-signature length cannot be its norm without the subspace selection of Section~\ref{sec:geodesicnorm}. The coefficient entering $\mathcal Y$ acquires a white-noise weight as soon as its velocities are mixed by a continuum through a coordinate-dependent coupling (Proposition~\ref{prop:bath}), and the programme's own texture supplies such a continuum in the thermodynamic limit at fixed spacing (Example~\ref{ex:sutherland}); what \eqref{eq:push} still needs is the projection $\mathcal R$---the positive effective state space, which by Lemma~\ref{lem:positive} is a reconstruction with a changed product and a sector parameter, the collective alternatives, the projected ASA generator with its treatment of $b$, the counter-term and the response gain, and the preparation probabilities---evaluated for one interacting sector. Once these are fixed, the probabilities cannot be reweighted to obtain the desired answer.

The practical sequence: choose one interacting sector and a normalisable preparation; the Sutherland texture of Example~\ref{ex:sutherland} is the natural choice, since its spectrum is computed and the coordinate-dependent coupling it needs is of the type of the vertex \eqref{eq:vertex}. Construct its physical projection and derive the projected ASA term together with the response, carrying the sea's reaction along: compute the symmetrised and retarded coefficient kernels, as the audit of Ref.~\cite{audit} also demands, and read off $\varkappa$ and $\alpha$ from them by Propositions~\ref{prop:history} and \ref{prop:einstein}---a derivation that returns $\alpha\neq1$ or $\varkappa\neq1$ has produced a law that signals or selects with the wrong frequencies, and is the most informative negative result available. If a geometric length is used for normalisation, identify the positive quantity that becomes $\|\psi\|^2$ and show that its conservation survives the coupled dynamics. Compute the joint coarse-grained law and test it against \eqref{eq:generator}, or prove mean Born preservation and completed selection directly; ensembles with the same density matrix must give the same averaged law. Then reconstruct the classical description and allowed operations, and apply the causal and experimental tests, including coherent response terms and existing collapse bounds \cite{bassi}.

Table~\ref{tab:status} collects the answers. The picture with which we began survives with its burden relocated. Selection in Connes time defers the hypersurface question, and the effective law constructed here, once its records are read in emergent spacetime, is non-signalling by the standard theorem; the protection holds only for laws that are linear on ensembles and composed of complete local operations, the obvious shortcuts lie outside that class, the state dependence inside it is confined to the unravelling and cannot raise the rate, and for a history produced by a sea the same protection fixes the sea's response gain to its noise strength. On the microscopic side the programme's own texture, taken long at fixed spacing, gives the frozen coefficient an integrable correlation, a computable white-noise strength and an accompanying friction, once a coordinate-dependent coupling to it is granted, and the block that carries the coefficient is a reconstructed one on which the kinetic form is positive. What remains is the projection: the reconstruction and its sector, the coupling, the counter-term and the gain. The question is no longer whether a suitable collapse equation can be written; it is whether the interacting GTD dynamics, its preparation and its physical projection produce that equation---and the two numbers that decide it, $\varkappa$ and $\alpha$, are now things a derivation can be asked to return.

\FloatBarrier
\appendix
\section{The microscopic action and two projection tests}\label{app:micro}
\subsection{The action}
With $\eta=L_{\rm Pl}^2/L^2$ the ratio of the Planck length to the regulator length, $\nu=\alpha/L$ a frequency fixed by a dimensionless constant $\alpha$, and $D_\pm=\partial_\tau\pm i\nu$, the action considered in the monograph \cite[Sec.~7.3.6]{book} has the form
\begin{equation}
\mathcal L=c_0\eta\Tr[(D_-q_1^\dagger)(D_+q_2)],\qquad
q_1^\dagger=q_B^\dagger+\eta\beta_1q_F^\dagger,\qquad q_2=q_B+\eta\beta_2q_F .
\label{eq:action}
\end{equation}
The independent variables are the fundamental graded coordinates and their adjoints; the auxiliary Grassmann generators are not inverted. On the restricted domain the free equations include $D_+^2q_B=D_+^2q_F=0$. The real scalar projection adopted there supplies a regulated canonical Hamiltonian but not a separately derived ASA collapse generator; an ASA bifermionic Lagrangian coefficient is not automatically such a generator, and the effective equation $i\hbar\dot\psi=(H_S-i\hbar G_F)\psi$ remains conditional on an additional projection, which supersedes the earlier identification $H=\mathcal L$ in the collapse discussion \cite{gtdborn}. The free regulated model has noncompact flat directions, so an invariant coefficient measure is not automatically a normalisable ensemble; the microscopic calculation needs an explicit preparation, a system--sea split and a physical observable map \cite[Secs.~9.2--9.4]{book}.

\subsection{A positivity restriction on the auxiliary map}
For independent anticommuting square-zero generators of matching adjoint parity, $b=\beta_1\beta_2$ satisfies $b^\dagger=-b$, $b^2=0$ and $b^\dagger b=0$. For any normalised positive linear functional $\omega$ on this star algebra, Cauchy--Schwarz gives $|\omega(b)|^2\leq\omega(1)\omega(b^\dagger b)=0$, so a nonzero replacement $b\mapsto i\rho_0$ is not a positive expectation on that algebra (the Gram matrix of $1,b$ would have determinant $-\rho_0^2$), and a finite-dimensional star representation on a positive Hilbert space sends $b$ to zero. This does not prohibit an emergent change of product; it says that a positive physical state space cannot be justified by interpreting the replacement as ordinary averaging on the original auxiliary algebra.

\subsection{What the proposed vertex can distinguish}
The proposed coassociative interaction has the current--current form \cite{cyclic}
\begin{equation}
\mathcal V_4=\kappa\,\varphi_{abcd}\Tr(\mathcal J^{ij}_{ab}\mathcal J^{kl}_{cd}),\qquad \mathcal J^{ij}_{ab}=q_i^aq_j^b ,
\label{eq:vertex}
\end{equation}
with $\varphi_{abcd}$ the totally antisymmetric coassociative four-form, $a,\ldots,d$ octonionic indices and $i,\ldots,l$ labelling the coordinates. Total antisymmetry projects onto two-form currents, so $\varphi_{abcd}\delta_{cd}=0$ and a scalar-density bath contraction does not furnish a collapse coupling through this vertex. With independent antisymmetric index pairs, $\mathcal M_{(ab),(cd)}=\varphi_{abcd}$ satisfies $\mathcal M=2P_7-P_{14}$ in the real-octonion convention of that paper; if a projection supplied a real bath covariance $\mathcal N\succeq0$ and an effective coupling $g$, the induced covariance $g^2\mathcal M\mathcal N\mathcal M^T$ would be positive, opposite interaction signs are not negative variances, and the ratio $4{:}1$ holds only for an isotropic covariance. A pure associator mechanism vanishes if all factors entering it lie in one associative subalgebra. Freezing two factors of \eqref{eq:vertex} to backgrounds $B^c,B^d$ leaves the kernel $\varphi_{abcd}B^cB^d=\frac12\varphi_{abcd}[B^c,B^d]$, which vanishes for commuting backgrounds; a bilinear coupling of the retained block to the remaining coordinates therefore requires noncommuting backgrounds, and its spectral density is part of the specification of the interacting sector.

\section{A positive fermionic reservoir}\label{app:bath}
Let $\{c_a,c_b^\dagger\}=\delta_{ab}$, $H_B=\sum_{a=1}^N\epsilon_ac_a^\dagger c_a$ with $\epsilon_a=E_W[(a-\frac12)/N-\frac12]$ for a bandwidth $E_W$, $d=N^{-1/2}\sum_ac_a$ and $B=v(d^\dagger d-\frac12)$, a bounded self-adjoint bilinear ($B=(iv/2)\gamma_1\gamma_2$ in Majorana form, a Clifford product and not the square-zero auxiliary product). In the tracial preparation $\rho_B=\id/2^N$,
\begin{equation}
G_N(t)=\Tr[\rho_BB(t)B(0)]=\frac{v^2}{4N^2}\Bigl|\sum_ae^{-i\epsilon_at/\hbar}\Bigr|^2
\;\xrightarrow{N\to\infty}\;
G(t)=\frac{v^2}4\sinc^2\frac{E_Wt}{2\hbar},
\label{eq:G}
\end{equation}
with zero-frequency weight $D_B=\int_{-\infty}^\infty G\,\dd t=\pi\hbar v^2/(2E_W)$ and spectrum $S(\omega)=D_B(1-\hbar|\omega|/E_W)_+$: particle--hole energy differences fill a continuum and no linewidth is inserted. At fixed $N$ a constant component $v^2/4N$ remains and recurrences occur on the scale $2\pi\hbar N/E_W$. At temperature $T$ and zero chemical potential $D_B(T)=2\pi\hbar v^2T\tanh(E_W/4T)/E_W^2$, which tends to $D_B$ at large $T$ and vanishes linearly as $T\to0$; the retarded susceptibility $\chi_B(t)=-\frac2\hbar\theta(t)\operatorname{Im}G_T(t)$ vanishes in the tracial limit while the noise does not.

For $Q=\sum_jq_jP_j$ with $[H_S,Q]=0$ and $H=H_S+H_B+\lambda Q\otimes B$, the second-order Markov approximation gives $\dot\rho=-\frac{\Gamma_{\rm d}}2[Q,[Q,\rho]]$ with $\Gamma_{\rm d}=\lambda^2D_B(T)/\hbar^2$. An exact check uses the fermionic trace--determinant identity \cite{klich}: with $u_a=N^{-1/2}$ and $h_j=\operatorname{diag}(\epsilon_a)+\lambda q_jv|u\rangle\langle u|$, a coherence is multiplied by $F_{jk}(t)=e^{i\lambda v(q_j-q_k)t/2\hbar}\det[(\id+e^{ih_kt/\hbar}e^{-ih_jt/\hbar})/2]$; for $E_W=4$, $v=\hbar=1$, $q_j-q_k=2$ and the pre-recurrence window $t\in[10,30]$ the exact logarithmic decay slopes are $0.03164$ ($N=96$) and $0.03163$ ($N=192$) at $\lambda=0.2$ against $2\lambda^2D_B=0.03142$, and $0.007819$, $0.007816$ at $\lambda=0.1$ against $0.007854$. Since $[P_j\otimes\id,H]=0$ the exact evolution preserves every record-sector weight, and in the tracial preparation $U_j\rho_BU_j^\dagger=\rho_B$ for every $q_j$, so a final bath measurement cannot distinguish records: the reservoir dephases and does not select. The second-order coherent response $H_{\rm LS}\propto\operatorname{Im}\int_0^\infty G_T\,Q^2$ vanishes for the real tracial correlation but not at finite temperature; for $Q=Q_A+Q_B$ its cross term $Q_AQ_B$ must be given an admissible emergent support before a common bath is used as a model for spacelike laboratories.

\section{Probability details of the joint law}\label{app:joint}
\subsection{Why the norm weight appears}
Suppose a linear completely positive operation density $\mathcal I_y$ sends every pure input to the ray of an invertible $M_y\psi$, with Kraus operators $A_{yr}$. Purity forces each $A_{yr}\psi$ to be parallel to $M_y\psi$, so $M_y^{-1}A_{yr}$ maps every vector into its own span and is a multiple of the identity. Hence $\mathcal I_y(\rho)=r(y)M_y\rho M_y^\dagger$ with $r\geq0$, and the probability is $r(y)\|M_y\psi\|^2$. Once the instrument assumptions are made, updates and probabilities are not independent choices; the trace rule and affine dependence on preparations remain assumptions.

\subsection{The joint generator and consistency in time}
Acting on a smooth $F(p,Y)$, the joint generator of \eqref{eq:Y}--\eqref{eq:p} is
\begin{equation}
\mathcal AF=2\bar\ell\cdot\nabla_YF+\tfrac12\Delta_YF+\sum_{ja}b_{ja}\partial_{p_j}\partial_{Y_a}F+\tfrac12\sum_{jk}(bb^T)_{jk}\partial_{p_j}\partial_{p_k}F .
\label{eq:generator}
\end{equation}
The cross derivatives encode the correlation between histories and state changes; a noise spectrum alone does not fix them. For constant commuting channels the reference likelihood $Z_\tau=\sum_jp_j^0e^{2\ell_j\cdot Y_\tau-2|\ell_j|^2\tau}$ is a $\mathbb Q$-martingale, so \eqref{eq:measure} defines consistent finite-time laws without any assertion about an infinite-time absolutely continuous measure; with separate local exposure clocks the same holds on the product filtration.

\subsection{Finite-exposure error bound}\label{app:bound}
Let $q$ be the product reference Gaussian density, $\Lambda_j=e^{2y_j-2s_j}$ the likelihood ratio of record $j$ from \eqref{eq:localK}, $a_j=p_j^0\Lambda_j$ and $Z=\sum_ja_j$, so that the physical density is $Zq$ and $p_j=a_j/Z$. Then
\[
\E_{\mathbb P}\Bigl(1-\sum_jp_j^2\Bigr)=\sum_{j<k}\E_{\mathbb Q}\frac{2a_ja_k}{Z}\leq\sum_{j<k}\E_{\mathbb Q}\sqrt{a_ja_k},
\]
using $Z\geq a_j+a_k\geq2\sqrt{a_ja_k}$. Gaussian integration gives $\E_{\mathbb Q}\sqrt{\Lambda_j\Lambda_k}=e^{-(s_j+s_k)/2}$, and $\sum_jp_j^2\leq\max_jp_j$, which proves \eqref{eq:error}. For general dimensionless record operators the exponent is $D_{jk}/2$ with $D_{jk}=\sum_is_i(v_{ij}-v_{ik})^2$.

\section{The linearised-bath model}\label{app:bathmodel}
\subsection{Model and exact reduction}
Each complex block coordinate $Q_r$ has two real components, which the coupling below does not mix; it suffices to treat one real coordinate $q$ of mass $m$ with velocity $v=\dot q$, and to assemble $v_r=v_{r,x}+iv_{r,y}$ from two independent identical copies at the end. Couple $q$ to $N$ oscillators $z_k$ of frequency $\Omega_k$ through
\begin{equation}
H=\frac{p^2}{2m}+\sum_{k=1}^N\Bigl[\frac{\pi_k^2}{2}+\frac{\Omega_k^2}{2}\Bigl(z_k-\frac{c_k}{\Omega_k^2}q\Bigr)^2\Bigr],
\label{eq:CL}
\end{equation}
the Caldeira--Leggett form \cite{caldeira}, in which the counter-term $\sum_kc_k^2q^2/2\Omega_k^2$ keeps the retained coordinate translation invariant, as the free GTD block is. In the relative coordinates $u_k=z_k-c_kq/\Omega_k^2$ the equations close on $(p,u_k,\pi_k)$: $\dot p=\sum_kc_ku_k$, $\dot u_k=\pi_k-c_kp/m\Omega_k^2$, $\dot\pi_k=-\Omega_k^2u_k$. Eliminating the bath gives the exact generalised Langevin equation \cite{fordkacmazur,zwanzig}
\begin{equation}
m\dot v(t)=-\int_0^tK(t-s)v(s)\dd s+\xi(t),\qquad
K(t)=\sum_k\frac{c_k^2}{\Omega_k^2}\cos\Omega_kt,
\end{equation}
with $\xi$ a linear functional of the bath initial data. In the stationary Gaussian ensemble of \eqref{eq:CL} at temperature $T$---$p\sim\mathcal N(0,mT)$, $\pi_k\sim\mathcal N(0,T)$, $u_k\sim\mathcal N(0,T/\Omega_k^2)$, all independent, which is the Gibbs measure since the change $z_k\to u_k$ at fixed $q$ has unit Jacobian---$\xi$ is Gaussian with $\langle\xi(t)\xi(s)\rangle=TK(t-s)$, the classical fluctuation--dissipation relation, and $v$ is a stationary Gaussian process with $\langle v(t)v(0)\rangle=(T/m)[e^{\mathsf At}]_{pp}$, where $\mathsf A$ is the $(2N+1)$-dimensional linear generator. This is the formula evaluated in Figure~\ref{fig:bath}. The counter-term makes $P_{\rm tot}=p+\sum_k(c_k/\Omega_k^2)\pi_k$ conserved, so at every finite $N$ the velocity correlation has a constant component $T/M_{\rm eff}$ with $M_{\rm eff}=m+\sum_kc_k^2/\Omega_k^4$; in the Ohmic discretisation below $M_{\rm eff}\simeq m(1+\pi\gamma N/\omega_c)$, which tends to infinity in the continuum limit and removes the zero-frequency line.

\subsection{Ohmic discretisation and the Markov limit}
With spectral density $\varrho(\omega)=\frac\pi2\sum_k(c_k^2/\Omega_k)\delta(\omega-\Omega_k)=\gamma m\omega$ for $\omega<\omega_c$, discretised as $\Omega_k=(k-\frac12)\omega_c/N$ and $c_k^2=\frac2\pi\gamma m\Omega_k^2\omega_c/N$, the kernel is $K(t)=\frac{2\gamma m}{\pi}\sin(\omega_ct)/t\to2\gamma m\delta(t)$ as $\omega_c\to\infty$, the friction becomes $m\dot v=-\gamma mv+\xi$, and $v$ is an Ornstein--Uhlenbeck process with $\langle v(t)v(0)\rangle=(T/m)e^{-\gamma|t|}$. Assembling the complex velocity from two independent copies, $C_r(t)=2(T/m)e^{-\gamma|t|}$, so $\sigma_r^2=2T/m$, circularity $\langle v_r(t)v_r(s)\rangle=0$ holds for all time pairs, and $\langle J(t)J(0)\rangle=8\,c_1(t)c_2(t)$ in terms of the real-component correlations $c_r$. The ancillary script checks the Wick formula by Monte Carlo on independent Ornstein--Uhlenbeck processes, the exact finite-$N$ correlation against $e^{-\gamma t}$ ($N=64,256,1024$; maximal deviation before recurrence $0.020$, $0.010$, $0.005$), the variance of $Y(\tau)$ against \eqref{eq:ohmic} ($0.7\%$ at $N=1024$, $1\%$ at $N=512$), the spectral formula for $D_J$ against the closed form, the ballistic late-time growth of the $N=32$ bath, and the vanishing of $J$ for identical components; a separate script computes the spectrum of the linear generator for a negative-kinetic block coupled to a positive or a negative bath sector, the statement about $(2,2)$ forms recorded after Lemma~\ref{lem:positive}. For a super-Ohmic density $\varrho\propto\omega^{\mathsf s}$ the low-frequency velocity spectrum behaves as $\omega^{1-\mathsf s}$ for $1<\mathsf s<2$, so $S_1S_2$ is integrable only for $\mathsf s<\frac32$, while for $\mathsf s>2$ a finite effective mass restores a zero-frequency line; the Ohmic case is the simplest one with a Markov limit, and the general statement is the integrability condition of Proposition~\ref{prop:bath}. If the linearisation leaves a residual quadratic term for the block instead of cancelling it, the block velocity becomes that of a damped oscillator; its spectrum is then absolutely continuous with $S_r(0)=0$, since a confined coordinate has no velocity diffusion, but $\int S_1S_2\,\dd\omega$ remains positive, so $D_J>0$ still: the coefficient noise does not require the block itself to diffuse. The script evaluates this case as well.

\subsection{The Sutherland sea}\label{app:sutherland}
For the chain of Example~\ref{ex:sutherland}, with $\theta_d=\pi d/M$ and $k(\theta)=\mathsf g(\pi/\mathcal L)^4(6\csc^4\theta-4\csc^2\theta)$ the second derivative of the pair potential, the Hessian is the circulant $K_{ij}=-k(\theta_{i-j})$ for $i\neq j$, $K_{ii}=\sum_{d\neq0}k(\theta_d)$, with eigenvalues $\mathsf k_\ell=\sum_{d=1}^{M-1}k(\theta_d)[1-\cos(2\pi\ell d/M)]=4\mathsf g(\pi/\mathcal L)^4\ell^2(M-\ell)^2$, the result quoted in Ref.~\cite{projectvii}; the ancillary script confirms it to $10^{-8}$ for $M=16,64,256$. The frequencies \eqref{eq:sutherland_modes} follow with $\omega_\ell^2=\mathsf k_\ell/m$. In this subsection $m$ is the mass of a chain particle and $m_b$ that of the block coordinate (the mass called $m$ in the two preceding subsections). The real normal modes are $\sqrt{2/M}\cos(2\pi\ell n/M)$ and $\sqrt{2/M}\sin(2\pi\ell n/M)$ for $\ell=1,\ldots,M/2-1$, degenerate in pairs, together with the single mode $M^{-1/2}(-1)^n$ at $\ell=M/2$; the strain $x_1-x_0$ has amplitude $e_\ell$ on each, with $e_{\ell,c}^2+e_{\ell,s}^2=(8/M)\sin^2(\pi\ell/M)$ for a pair and $4/M$ for the edge mode. Writing the block Hamiltonian as \eqref{eq:CL} with $\Omega_k\to\omega_\ell$ and $c_k\to ce_\ell/\sqrt m$, the spectral density $\varrho(\omega)=\frac\pi2\sum_{\text{real modes}}(c^2e_\ell^2/m\omega_\ell)\delta(\omega-\omega_\ell)$ has, for $\ell\ll M$ where $\omega_\ell\simeq(4\Omega_{\max}/M)\ell$ and $\dd\ell/\dd\omega=M/(4\Omega_{\max})$, the continuum value $\varrho(\omega)\simeq\frac\pi2\frac{c^2}{m}\frac{8\pi^2\omega}{(4\Omega_{\max})^3}$, which is \eqref{eq:sutherland_ohmic}; the per-mode values at $\ell=1,2,3$ for $M=2048$ agree with it to $0.2\%$, and the ratio $\varrho/(\gamma m_b\omega)$ is $1.05$, $1.3$, $1.8$, $2.9$, $7.1$ at $\omega/\Omega_{\max}=0.05$, $0.25$, $0.5$, $0.75$, $0.95$. The sum in \eqref{eq:sutherland_meff}, over the pairs $\ell=1,\ldots,M/2$ with $\sum_\ell\ell^{-2}\to\pi^2/6$, gives $M_{\rm eff}-m_b\simeq\pi\gamma m_bM/(12\Omega_{\max})$; the exact sums at $M=1024$ and $4096$ agree with it to $3\%$. The exact velocity correlation is computed from the normal modes of the coupled quadratic system, $\langle v(t)v(0)\rangle=T\sum_nU_{qn}^2\cos\Omega_nt$ with $U^T\mathsf{Mass}\,U=\id$, which places the zero-frequency line at $TU_{q0}^2=T/M_{\rm eff}$ exactly; the kinetic-coupling case uses the same formula with the mass matrix $\bigl(\begin{smallmatrix}m_b&k_ve^T\\k_ve&m\id\end{smallmatrix}\bigr)$, $e=(-1,1,0,\ldots)$, and its zero mode carries weight $T/m_b$ exactly, as the conservation of $p_q$ requires. The colour sector of the texture, with its $r^2$ zero modes and its own bands \cite[eqs.~47--48]{projectvii}, is not included; it would add modes, not remove the positional ones.

\subsection{Scope}
The model is the coefficient-level consequence of a positive block kinetic term, a bilinear coordinate coupling and a stationary Gaussian preparation; by Lemma~\ref{lem:positive} the positive kinetic term is that of the reconstructed block, not a departure from it. Which factors of \eqref{eq:vertex} are frozen to produce the coupling to the texture, the value of $c$, and the attachment to the strain are part of the specification of the interacting sector and are not derived here. The calculation says nothing about the projection of $J$ to an operator on the effective state space, which is the open step identified in Section~\ref{sec:projection_result}, nor about the response gain of Proposition~\ref{prop:einstein}, which requires the back-reaction of the block on the sea to be specified.

\par\medskip\noindent\begin{minipage}{\textwidth}
\section*{Reproducibility and computational assistance}
The ancillary files contain finite Gaussian-law checks, local-operation tests, the entanglement-rate counterexample, the selection-rate bound for state-dependent unravellings, the stability check for an indefinite block, split-biquaternion and exterior-algebra checks, the exact free-coefficient history, the exact solution of the bare law and its comparison with the CSL law, the history-form drift formula and the fractional exit law with the Ornstein--Uhlenbeck sea of Proposition~\ref{prop:einstein}, the linearised-bath calculation with its finite-$N$ checks, the Sutherland-texture sea with its spectrum, finite-$M$ dynamics, zero-frequency line and kinetic-coupling case, reservoir checks, and the finite-exposure Bell calculation. None of the scripts simulates an interacting GTD derivation of collapse; the scope of each calculation is stated where it is used.

OpenAI ChatGPT-6 Astra and Anthropic Claude Fable 5.1 were used in developing and critically examining the constructions, preparing verification code, and drafting and revising the exposition. Three independent adversarial review passes by the latter identified a misreading of Ref.~\cite{uniqueness}, the kinetic-term question that Lemma~\ref{lem:positive} resolves, the linear-space argument of Section~\ref{sec:geodesicnorm}, the record proviso of Proposition~\ref{prop:bound}, and the separate status of the counter-term in Section~\ref{sec:einstein}. Responsibility for the scientific content rests with the author.
\end{minipage}\par

\end{document}